\documentclass{article}

\usepackage[margin=1in]{geometry}

\usepackage{url}
\usepackage{hyperref}
\usepackage[utf8]{inputenc}
\usepackage[small]{caption}
\usepackage{graphicx}
\usepackage{amsmath, amssymb}
\usepackage{amsthm}
\usepackage{booktabs}
\usepackage{algorithm}
\usepackage{algorithmic}

\newtheorem{lemma}{Lemma}
\theoremstyle{definition}
\newtheorem{definition}{Definition}

\usepackage{xcolor,nicefrac,physics,thm-restate,tcolorbox,xspace}

\usepackage[capitalize]{cleveref}

\DeclareMathOperator*{\E}{\mathbb{E}}

\newcommand{\Bern}{\mathrm{Bernoulli}}
\newcommand{\strategy}{\hyperref[str:RRB]{\color{black}Randomized Robust Bidding}\xspace}
\newcommand{\RRB}{\hyperref[str:RRB]{\color{black}RRB}\xspace}

\author{
David X. Lin$^1$\footnote{Equal contribution, order randomized.}\thanks{Supported by NSF grant ECCS-1847393.}
\and
Daniel Hall$^1$\footnotemark[1]
\and
Giannis Fikioris$^1$\thanks{Supported in part by the Google PhD Fellowship, the Onassis Foundation -- Scholarship ID: F ZS 068-1/2022-2023, and ONR MURI grant N000142412742.}\and
Siddhartha Banerjee$^1$\thanks{Supported in part by AFOSR grant FA9550-23-1-0068, ARO MURI grant W911NF-19-1-0217, and NSF grants ECCS-1847393 and CNS-195599.}\and
\'Eva Tardos$^1$\thanks{Supported in part by AFOSR grant FA9550-23-1-0410, AFOSR grant FA9550-231-0068, and ONR MURI grant N000142412742.}
\\[10pt]
$^1$Cornell University\\[10pt]
\{deh275, dxl2\}@cornell.edu,
gfikioris@cs.cornell.edu,
\{sbanerjee, eva.tardos\}@cornell.edu
}

\newif\ifarxiv
\arxivtrue

\title{Online Resource Sharing: Better Robust Guarantees via Randomized Strategies}
\begin{document}
\maketitle

\begin{abstract}
    We study the problem of fair online resource allocation via non-monetary mechanisms, where multiple agents repeatedly share a resource without monetary transfers.
Previous work has shown that every agent can guarantee $1/2$ of their ideal utility (the highest achievable utility given their fair share of resources) robustly, i.e., under arbitrary behavior by the other agents.
While this $1/2$-robustness guarantee has now been established under very different mechanisms, including pseudo-markets and dynamic max-min allocation, improving on it has appeared difficult.

In this work, we obtain the first significant improvement on the robustness of online resource sharing. In more detail, we consider the widely-studied repeated first-price auction with artificial currencies.
Our main contribution is to show that a simple randomized bidding strategy can guarantee each agent a $2 - \sqrt 2 \approx 0.59$ fraction of her ideal utility, irrespective of others' bids.
Specifically, our strategy requires each agent with fair share $\alpha$ to use a uniformly distributed bid whenever her value is in the top $\alpha$-quantile of her value distribution. Our work almost closes the gap to the known $1 - 1/e \approx 0.63$ hardness for robust resource sharing; we also show that any static (i.e., budget independent) bidding policy cannot guarantee more than a $0.6$-fraction of the ideal utility, showing our technique is almost tight.
\end{abstract}

\section{Introduction}

There are many settings where multiple self-interested agents share a resource controlled by a principal.
Specifically, we consider the problem where the resource is repeatedly allocated to one of the agents over a long time horizon.
For example, consider multiple scientific labs at a university sharing a computer cluster.
Each lab wants to run experiments in every time slot, but only one lab can use the cluster at a time.
The system administrator must decide who uses the cluster in each time slot.
This has to be done in a ``fair'' way so that each lab is satisfied with the resulting allocation.
In addition, the allocation has to happen without requiring monetary payments from the labs.

The allocation model and mechanism that we study were first introduced by \cite{gorokh2019remarkable} and subsequently also studied in \cite{fikioris2023online,banerjee2023robust}.
Every agent is endowed with a \textit{fair share} of the resource.
Roughly speaking, an agent's fair share is their intrinsic right to enjoy a fraction of the resource.
Letting different agents have different fair shares encapsulates many realistic real-world scenarios.
For example, it may be natural for a larger research group within a university to get access to a computer cluster more often than a smaller group.
As a benchmark, this line of work defines and uses \textit{ideal utility}.
Roughly speaking, an agent's ideal utility is the maximum per-round utility that she can obtain if she is restricted to obtaining only her fair share of the resource.

\cite{gorokh2019remarkable} introduce a simple non-monetary mechanism in which each agent is endowed with an amount of artificial currency proportional to their fair share, and at each time, the item is allocated according to a first-price auction using the artificial currency.
Assuming that each agent's value for the item is drawn independently from a fixed distribution each round, they show that in this mechanism, each agent can guarantee a $\frac 1 2-o(1)$ fraction of her ideal utility robustly, which means that each agent can make this guarantee regardless of the behavior of the other agents and their value distribution.
\cite{fikioris2023online} study the same problem under a different mechanism and get the same result.

\paragraph{Our results}
We use the model and mechanism of \cite{gorokh2019remarkable} and offer a much more detailed analysis of the achievable robust guarantees.
Our main result is an improved bidding strategy that any agent can follow to guarantee a $2-\sqrt2-o(1)\approx 0.59$ fraction of her ideal utility robustly (\cref{thm:lowerBound}).
We emphasize that since this result holds under arbitrary behavior by other agents, the same utility guarantee can be made under any equilibrium.
This lower bound is close to the upper bound of $1 - \frac{1}{e} \approx 0.63$ that \cite{fikioris2023online} give.
Our proposed strategy, \textit{\strategy}, is very simple.
While budget suffices, an agent's bid is sampled from a certain uniform distribution if her value is in the top $\alpha$-quantile of her value distribution, where $\alpha$ is her fair share.

This randomization in bidding is essential for improving the $\frac 1 2$ ideal utility guarantee.
For such a guarantee \cite{banerjee2023robust,gorokh2019remarkable} suggest that an agent should bid a fixed constant whenever her value is sufficiently high.
Under such simplistic bidding, the other agents know exactly how much they need to bid to beat agent $i$.
It is not hard to show that no fixed bidding can robustly guarantee a better than $\frac 1 2$ of the ideal utility.
The same holds for the $\frac 1 2$ guarantee of \cite{fikioris2023online}: their much simpler mechanism only allows requesting the resource or not, equating to a fixed bid for pseudo-markets.

Given the above observations about fixed bidding, a more involved strategy is essential for stronger results.
While a uniform distribution is arguably the simplest continuous distribution, it enjoys the following property.
It is the distribution that minimizes spending, subject to achieving the same fraction of ideal utility when the other agents are bidding $b$, for a range of values $b \ge 1$.
Carefully optimizing this uniform distribution, we get the $(2-\sqrt 2)$ guarantee of \cref{thm:lowerBound} when its support is $[0, 1 + \sqrt 2]$.

A main component of our utility lower bound (and subsequent results) is \cref{lem:bernoullivaluedistributionisworstcase}, which greatly simplifies the problem of getting robust guarantees.
Specifically, we show that for any $\beta \le 1$, if an agent can guarantee a $\beta$ fraction of her ideal utility when her value distribution is Bernoulli with mean equal to her fair share, then she can guarantee the same for any value distribution.
In addition, our reduction works for any mechanism, showing that this Bernoulli value distribution is the worst-case for this problem.

We show that our robust utility guarantee is almost tight in the following sense.
Suppose the agent does have the worst-case Bernoulli distribution.
If the agent uses any strategy that involves bidding from the same fixed distribution every round in which she has value $1$, she cannot robustly guarantee more than a $0.6$ fraction of their ideal utility (\cref{thm:fixeddistributionupperbound}).
This result also showcases the strength of the uniform distribution we used before.
If the probability that the agent bids more than $b$ when she has value $1$ is not at least $(1 - \beta) b$ at every point $b\geq 1$ (i.e., above a certain uniform distribution), then the other agents can bid $b$ to make the agent get less than $\beta$ fraction of her ideal utility.
In other words, we prove that if the agent's bidding is not aggressive enough, she cannot guarantee a $\beta$ fraction of her ideal utility.
On the other hand, by bidding too much, the agent will run out of budget.
We combine these bounds to show that an agent cannot obtain a guarantee of more than a $0.6$ fraction of their ideal utility regardless of the choice of bidding distribution.

Our final theoretical result is a bidding strategy the agents other than $i$ can follow that ensures agent $i$ receives at most a $(1-\frac 1 e)$ fraction of her ideal utility under the previous worst-case Bernoulli value distribution (\cref{thm:arbitrarystrategyupperbound}).
Unlike the previous result, this is under arbitrary strategies that agent $i$ can follow, such as time-varying strategies.
Due to \cite{fikioris2023online}, it is already known that there is no mechanism can guarantee every agent a $(1 - \frac 1 e)$ fraction of her ideal utility but this is the first explicit strategy that does this.
We believe that the bidding distribution we provide for the other agents might be of interest to close the gap between the upper and lower bounds.

Finally, in \cref{sec:experiments}, we provide empirical evidence of how our \strategy strategy is superior to the fixed bidding of previous work.
In particular, we show that when each of $n$ agents with equal fair shares follows the \strategy strategy, then every agent ends up with a $1-(1-1/n)^n$ fraction of their ideal utility (note that $1-(1 - 1/n)^n \to 1 - 1/e$ as $n \to \infty$).
This fraction is the theoretical maximum \textit{any allocation} can guarantee when there are $n$ agents, even if it has knowledge of the agents' realized values.
This is much higher than the $1/2$ fraction of ideal utility that agents get when following the fixed bidding strategy of \cite{gorokh2017,banerjee2023robust}.

\paragraph{Related Work}

Our work follows in a long line of works that consider the problem of repeated allocation of resources without money.
The interest in studying such mechanisms stems from its application in many real-world settings, and indeed, the theory has benefited from and influenced successful deployments for course allocation~\cite{budish2017course}, food banks~\cite{walsh2014allocation,prendergast2022allocation} and cloud computing~\cite{dawson2013reserving,vasudevan2016customizable}.

The particular model we consider, with a single indivisible item per round, and agents with random valuations across rounds, was first considered in the work of~\cite{guo2010}; however the core idea of `linking' multiple allocations to incentivize truthful reporting without money goes back to the seminal work of~\cite{jackson}. 
These mechanisms, and a long line of follow-up work~\cite{cavallo2014incentive,gorokh2017,balseiro2019multiagent,blanchard2024near}, provide only Bayes-Nash equilibrium guarantees, and moreover, the mechanisms need to know the value distributions beforehand. None of these mechanisms, however, can provide any guarantee under non-equilibrium actions by other agents.

A more recent line of work, starting from \cite{gorokh2019remarkable}, considers the same setting, but focuses on \emph{robust individual-level} guarantees: the aim now is to guarantee each agent some minimum utility irrespective of how other agents behave. 
\cite{gorokh2019remarkable} propose the repeated first-price pseudo-market that we also use in our work, and show that every agent can guarantee a $1/2$ fraction of her ideal utility robustly.
Since then, this $1/2$-robustness guarantee has been re-obtained using very different mechanisms: \cite{banerjee2023robust} give a simple argument to get this guarantee using a repeated first-price auction with a reserve (their main focus is to extend the robustness guarantees to reusable resources, i.e., resources that an agent might want for multiple consecutive rounds), and
\cite{fikioris2023online} show how to get it using a very different non-market based mechanism called Dynamic Max-Min Fairness  (while also showing how to get robust guarantees that are distribution-specific and extend to values that are correlated across time).
Both these latter works also suggest that the $1/2$-robustness guarantee is essentially tight under their respective approaches, which is far below the best upper-bound of $1-1/e$ we discuss above. The question of whether one can obtain stronger robustness guarantees for the basic single-item setting has, however, remained open until this work.

Further afield, there are problems where agents' values are known in advance.
\cite{babaioff2021fair,babaioff2022best} study fair resource allocation when the agents' values are adversarially picked have to be allocated simultaneously.
Due to this information structure, their utility benchmarks are much weaker than ours.
For example, for an agent who has positive value $1$ for only $T/n$ out of $T$ items, they guarantee $\Theta(T/n^2)$ utility as an adversary can ``block'' the agent for only items of positive value.
In contrast, for an agent with Bernoulli$(1/n)$ values, we guarantee $\approx0.59 T/n$ utility.
\cite{avni2018infinite,lazarus1999combinatorial} study `Poorman games', which are games involving bidding with artificial currencies. However, unlike our setting, players have full-information, making them fundamentally different, similar to \cite{babaioff2021fair,babaioff2022best}.

\section{Preliminaries}

\subsection{Model and Ideal Utility}

We consider repeated online allocation of a single, indivisible resource via repeated first-price auctions using artificial currencies as introduced in \cite{gorokh2019remarkable}. There are $n$ agents, $1, 2, \ldots, n$. At each time $t=1,2,\dots,T$, a principal decides which agent, if any, to receive the resource. Each agent $i$ has a nonnegative value $V_i[t]$ for the item at time $t$ and aims to maximize her total utility, which is the sum of the values of the items she got allocated. We assume the values $V_i[t]$ are drawn independently across both agents $i$ and times $t$. Specifically, each agent has a time-invariant value distribution $\mathcal F_i$, and the values $V_i[t]$ are drawn from $\mathcal F_i$ each round $t$. We make no assumptions about $\mathcal F_i$ beyond non-negativity.
The values $V_i[t]$ are private and are not known to the other agents or the principal.

Each agent $i$ has some exogenously defined fair share $\alpha_i$, where each $\alpha_i\geq 0$ and $\sum_{i=1}^n\alpha_i=1$. An agent's fair share measures the exogenously defined fraction of allocated items they should receive in a fair world. A fair principal should consider mechanisms that favor agents with higher fair shares in some way.

As in \cite{gorokh2019remarkable,banerjee2023robust,fikioris2023online}, to evaluate an agent's resulting utility, we use the benchmark of \textit{ideal utility}. Intuitively, the ideal utility $v_i^\star$ of agent $i$ is the maximum long-term time-average utility the agent can get if allocated an $\alpha_i$ fraction of the rounds. Formally, $v_i^\star$ is defined as the maximum expected utility they could achieve from a single round if they can only obtain the item with probability at most $\alpha_i$:
\begin{equation}
\label{eq:idealutilitydefinition}
\begin{split}
    v_i^\star = \max_{\rho:[0,\infty)\to[0,1]}&\E_{V_i\sim\mathcal F_i}[V_i\rho(V_i)]\\
    \text{s.t.}\quad&\E_{V_i\sim\mathcal F_i}[\rho(V_i)]\leq \alpha_i
\end{split}
\end{equation}
where $\rho(V_i)$ denotes the probability of obtaining the item conditioned on the value $V_i$.

We will be interested in \textit{robust} strategies used by the agents. These strategies approximate an agent's ideal utility regardless of the other agents' behavior, even if they behave adversarially. We give a formal definition below.

\begin{definition} \label{def:robust_policy}
A policy used by an agent $i$ is \textit{$\beta_i$-robust} if when using the policy, regardless of the behavior of the other agents $j\neq i$, the agent's per-round expected utility is at least $\beta_i$ fraction of her ideal utility, i.e.,
\begin{equation*}
    \frac{1}{T}\sum_{t=1}^T\mathbb E[U_i[t]] \geq \beta_i v_i^\star.
\end{equation*}
\end{definition}

\paragraph{Robust Strategies, Equilibria and Price of Anarchy}
As also mentioned in previous work, we point out an additional benefit of \cref{def:robust_policy}.
If every agent has a $\beta$-robust policy, then under any equilibrium, every agent achieves a $\beta$ fraction of her ideal utility.
In addition, this implies Price of Anarchy guarantees when the agents' fair shares are equal, $\alpha_i = 1/n$.
In this case, the social welfare is upper-bounded by the sum of the agents' ideal utilities, implying that, if every agent gets an $\beta$ fraction of her ideal utility, the resulting social welfare is a $\beta$ fraction of the optimal one.
This means that the Price of Anarchy is at most $1/\beta$.



\subsection{Pseudo-market mechanism}

In this section, we introduce the mechanism we use to allocate the resource, a repeated first-price auction with artificial currency.
We note again that this is the same mechanism as proposed in \cite{gorokh2019remarkable}.

At the beginning of time, each agent is endowed with a budget of $B_i[1] = \alpha_i T$ credits of artificial currency. At each time $t$, each agent submits a bid $b_i[t]$ no more than their current budget $B_i[t]$. The principal selects the agent $i^\star$ with the highest bid (ties broken arbitrarily) to allocate the item to. The winning agent $i^\star$ pays her bid in artificial currency, and no other agent pays.
We denote the payment by agent $i$ as $P_i[t] = b_i[t]\pmb1\{i = i^\star\}$. The budgets get updated as $B_{i}[t+1] = B_{i}[t] - P_i[t]$.
The mechanism is summarized in Mechanism~\ref{alg:mechanism}.

\floatname{algorithm}{Mechanism}
\begin{algorithm}[tb]
    \caption{Repeated first-price auction with artificial currency}
    \label{alg:mechanism}
    \textbf{Input}: Number of rounds $T$, number of agents $n$, and fair shares $\alpha_1, \ldots, \alpha_n$
    \begin{algorithmic}[1] 
        \STATE Endow each agent $i$ with $B_i[1] = \alpha_i T$ tokens of artificial currency.
        \FOR{$t=1,2,\dots,T$}
        \STATE Agents submit bids $b_i[t]$ where each $b_i[t] \leq B_i[t]$.
        \STATE Select the winner $i^\star = \arg\max_i b_i[t]$ (ties broken arbitrarily).
        \STATE Set the payments as $P_i[t] = b_i[t]\pmb1\{i=i^\star\}$.
        \STATE Update budgets $B_i[t+1] = B_i[t] - P_i[t]$.
        \STATE Agents get utility $U_i[t] = V_i[t]\pmb1\{i=i^\star\}$.
        \ENDFOR
    \end{algorithmic}
\end{algorithm}

Agents have no intrinsic value for the artificial currency and simply aim to maximize their total received value.
Letting $W_i[t]$ be the indicator for whether agent $i$ won the item at time $t$, we denote an agent's utility gained at time $t$ by $U_i[t]$, defined as $U_i[t] = V_i[t]W_i[t]$. Each agent seeks to maximize their total utility, $\sum_{t=1}^T U_i[t]$.

\subsection{Formulation as a two-player zero-sum game}
\label{sec:formulationasatwoplayerzerosumgame}

When analyzing the robustness of a strategy for a particular agent $i$, we can think of the other $n-1$ agents as one combined adversary.
Specifically, we can think of a single player with budget the sum of the other agents' budget and her bids are the maximum of their bids.
Formally, at each time $t$, we let $B'[t] = \sum_{j\neq i}B_j[t]$ and $b'[t] = \max_{j\neq i}b_j[t]$.
From the perspective of agent $i$, participating in the mechanism is the same as playing against a single adversary whose budget and bids at round $t$ are $B'[t]$ and $b'[t]$.
This gives a reduction from the problem of competing against multiple other players to the problem of only competing against a single adversary, i.e., the one adversary can always simulate the behavior of $n-1$ other players.
Since we are studying bounds on the achievable robustness of agent $i$'s strategies, we can think of this as a zero-sum game: agent $i$ is trying to maximize her total utility $\sum_{t=1}^T U_i[t]$ and the other agents, thought of as a single combined adversary, is trying to minimize this.

This two-player zero-sum game is the perspective we will take in what follows. We will fix an agent $i$ with fair share $\alpha$ and refer to her simply as ``the agent'' or ``the player,'' and drop the $i$ subscript from our notation.
We will refer to the other players, playing as a single adversary, as just the ``adversary'' with a fair share $1 - \alpha$. In notation, we will use $'$s to denote the adversary's quantities, e.g., $B'[t]$ is the adversary's budget, $P'[t]$ is the adversary's payment, etc.

We will note that while we have a two-person zero-sum game formulation, this game is extremely complicated. The strategy space for each player is extremely large, including all possible history-dependent bidding policies over time. Therefore, it is difficult to analyze the equilibrium behavior of this two-person zero-sum game.

\section{Robust Strategy Lower Bound of \texorpdfstring{$2 - \sqrt 2$}{0.59}}
\label{sec:lower}

In this section, we give an $\big(2-\sqrt{2} - O(\sqrt{\nicefrac{\log T}{T}})\big)$-robust strategy for the agent.
This means an agent can guarantee approximately a $0.59$ fraction of their ideal utility under arbitrary behavior by the other agents.
This is the first robust bound that breaks the $1/2$ barrier and shows any agent can guarantee more than half her ideal utility regardless of the behavior of the other agents.

Before presenting our lower bound, we simplify the problem by a reduction.
Specifically, we reduce the problem of finding robust strategies for any value distribution $\mathcal F$ to finding robust strategies for the $\Bern(\alpha)$ values.
We prove that any $\beta$-robust strategy for the Bernoulli case can be converted to a $\beta$-robust strategy for arbitrary value distributions.

\subsection{Reduction to \texorpdfstring{$\mathrm{Bernoulli}(\alpha)$}{Bernoulli} value distributions}
\label{sec:reducingtobernoullivaluedistributions}

In this section, we reduce the problem of arbitrary value distributions to Bernoulli value distributions.
This reduction also implies that having a $\mathrm{Bernoulli}(\alpha)$ value distribution is the worst case for the agent trying to achieve a high robustness factor.
Our reduction is for an arbitrary mechanism, not just the pseudo-market we consider in the rest of the paper.

\begin{lemma}
\label{lem:bernoullivaluedistributionisworstcase}
Fix an arbitrary mechanism and an agent with fair share $\alpha$.
Assume there is a $\beta$-robust policy $\hat\pi$ for that agent when she has a $\hat{\mathcal F} = \mathrm{Bernoulli}(\alpha)$ value distribution.
Then for any value distribution $\mathcal F$, we can construct a $\beta$-robust policy $\pi$ for an agent with value distribution $\mathcal F$.
\end{lemma}

The lemma's proof is based on the following observation.
Using the definition of ideal utility, \cref{eq:idealutilitydefinition}, we can map the values $V[t]$ of any distribution $\mathcal F$ to Bernoulli values $\hat V[t]$ where $\hat V[t] = 1$ corresponds to the agent getting the resource when realizing her ideal utility.
The $\beta$-robust policy $\hat\pi$ can guarantee the item a $\beta$ fraction of the rounds when $\hat V[t] = 1$.
In other words, $\hat\pi$ can guarantee a $\beta$ fraction of the rounds that achieve $v^\star T$ utility in expectation.

\begin{proof}[Proof of Lemma \ref{lem:bernoullivaluedistributionisworstcase}]
Let $\rho^\star$ be the maximizer of \eqref{eq:idealutilitydefinition} with value distribution $\mathcal F$. We construct the policy $\pi$ to ``simulate'' $\hat\pi$. Let $\pi$ be the policy that, at each time $t$, defines $\hat V[t]$ to be $1$ with probability $\rho^\star(V[t])$ and $0$ otherwise, and behaves the same as $\hat \pi$ in the mechanism when using the values $\hat V[t]$. It samples the probabilities independently so that $\hat V[t]$ is independent of $\hat V[t']$ conditioned on $V[t]$ and $V[t']$ for $t\neq t'$.

By the feasibility and optimality of $\rho^\star$ in \eqref{eq:idealutilitydefinition}, the values $\hat V[t]$ are indeed i.i.d. $\mathrm{Bernoulli}(\alpha)$. Recall that $W[t]$ denotes the indicator that the agent wins the item in round $t$. Then,
\begin{equation}
\label{eq:pi_beta_robust}
    \sum_{t=1}^T \mathbb E[\hat V[t]W[t]] \geq \beta\alpha T
\end{equation}
by the $\beta$-robustness of $\hat\pi$. When the agent uses the strategy $\pi$, everything in the game is independent of the actual values $V[t]$ conditioned on the Bernoulli values $\hat V[t]$. Therefore,
\ifarxiv 
\begin{align*}
    \sum_{t=1}^T \mathbb E[V[t]W[t]] & = \sum_{t=1}^T \mathbb E[\mathbb E[V[t]W[t]\mid \hat V[t]]]\\
    & = \sum_{t=1}^T \mathbb E[\mathbb E[V[t]\mid\hat V[t]]\mathbb E[W[t]\mid\hat V[t]]]\\
    & =\sum_{t=1}^T \mathbb E[\mathbb E[\mathbb E[V[t]\mid \hat V[t]=1]\hat V[t]W[t]\mid \hat V[t]]]\\
    & =\sum_{t=1}^T\frac{v^\star}{\Pr(\hat V[t]=1)}\mathbb E[\hat V[t]W[t]]\\
    & =\frac{v^\star}{\alpha}\sum_{t=1}^T \mathbb E[\hat V[t]W[t]]\\
    & \geq \frac{v^\star}{\alpha}\beta\alpha T = \beta v^\star T.
\end{align*}
\else
\begin{align*}
    \sum_{t=1}^T & \mathbb E[V[t]W[t]] = \sum_{t=1}^T \mathbb E[\mathbb E[V[t]W[t]\mid \hat V[t]]]\\
    & = \sum_{t=1}^T \mathbb E[\mathbb E[V[t]\mid\hat V[t]]\mathbb E[W[t]\mid\hat V[t]]]\\
    & =\sum_{t=1}^T \mathbb E[\mathbb E[\mathbb E[V[t]\mid \hat V[t]=1]\hat V[t]W[t]\mid \hat V[t]]]\\
    & =\sum_{t=1}^T\frac{v^\star}{\Pr(\hat V[t]=1)}\mathbb E[\hat V[t]W[t]]\\
    & =\frac{v^\star}{\alpha}\sum_{t=1}^T \mathbb E[\hat V[t]W[t]]\\
    & \geq \frac{v^\star}{\alpha}\beta\alpha T = \beta v^\star T.
\end{align*}
\fi
Above, the second line follows from the independence of $V[t]$ and $W[t]$ conditioned on $\hat V[t]$. The third line follows from the fact that $\mathbb E[V[t]\mid\hat V[t]] = \mathbb E[V[t]\mid \hat V[t]=1]\hat V[t]$ since $\hat V[t]$ is Bernoulli. The fourth line uses $\mathbb E[V[t]\mid\hat V[t]=1] = v^\star/\Pr(\hat V[t]=1)$ which follows from the definition of ideal utility. The fifth line uses the fact that $\Pr(\hat V[t]=1)=\alpha$. The inequality comes from \eqref{eq:pi_beta_robust}.
\end{proof}

\subsection{The \texorpdfstring{\strategy}{Randomized Robust Bidding} strategy} \label{ssec:lower:strategy}

Assuming a $\mathrm{Bernoulli}(\alpha)$ value distribution, our proposed strategy for the agent is as follows. 
In rounds where the agent has a value of $1$, she bids using a sample from a specific uniform distribution (assuming she has enough budget).

We can see that using a proper distribution instead of a fixed bid (as done in previous work) is necessary to developing better than $1/2$-robust strategies as follows.
Assume the agent bids a fixed amount $b$ every round she has value $1$, while budget lasts.
If $b\geq 2$, then with a budget of $\alpha T$, she can only win at most $\alpha T/b$ times, yielding at most a $1/b \leq 1/2$ fraction of their ideal utility.
If $b\leq2$, then the adversary can bid $b+\epsilon$ to block the agent for approximately $T/b$ rounds, leaving only $(1 - 1/b)T \leq T/2$ rounds for the agent.

Contrarily, against a continuous bid distribution, the adversary cannot predict the agent's exact bid.
When both the agent and the adversary bid, the adversary has to either bid too high to ensure the agent does not get the item or bid lower but allow the agent to get the item with some probability.

We now describe our new Randomized Robust Bidding strategy for an arbitrary value distribution: the basic idea is that agent $i$ bids \emph{whenever} she has a high value (irrespective of everything else), but in a way such that her total bidding probability is exactly $\alpha_i$.
We formalize this below for the case where the value distribution of agent $i$ has a unique $\alpha_i$-quantile, as is the case for any absolutely continuous distribution; if the value distribution has atoms, then we can appropriately randomize at the cutoff to make sure the agent bids with probability exactly $\alpha_i$.

\begin{tcolorbox}[label=str:RRB,standard jigsaw,opacityback=0]
    \textbf{Randomized Robust Bidding (RRB)}\\
    If $V[t]$ is in the top $\alpha$-quantile of the value distribution, bid $b[t] \sim \text{Uniform}([0, \bar b])$, subject to the remaining budget.
\end{tcolorbox}

More formally, the \strategy proceeds as follows:
Let $\tau$ be the the last round in which the agent has at least $\bar b$ budget left, $\tau = \max\{t\geq 1:B[t] \geq \bar b\}$.
Let $\rho^\star$ be the optimal solution that realized the ideal utility in \eqref{eq:idealutilitydefinition}.
For each time $t$, if $t \leq \tau$, sample $\hat V[t]\sim\Bern(\rho^\star(V[t]))$.
If $\hat V[t]=1$, bid $b[t]\sim\mathrm{Uniform}([0,\bar b])$.

\begin{restatable}{theorem}{lowerBound} \label{thm:lowerBound}
The \RRB strategy with $\bar b = 1+\sqrt2$ is $\beta$-robust for
\begin{equation*}
    \beta = \left(2-\sqrt{2} - O\left(\sqrt{\frac{\log T}{T}}\right)\right)
\end{equation*}
for any distribution $\mathcal F$ that the agent has.
\end{restatable}

The uniform bidding distribution used in the \RRB strategy stems from the following simplified analysis, which we also use for an upper bound in \cref{ssec:upper:fixed_distr} (see \cref{lem:cdfrestriction}).
Consider the agent's bids are sampled from a distribution with CDF $F(\cdot)$.
Also assume that her fair share $\alpha$ is small.
Since the agent's behavior is the same round, consider the adversary also acting the same every round.
Specifically, consider the adversary bidding $b'$, where $b' \ge 1$ (as otherwise she would have leftover budget).
While the adversary's budget lasts, they pay $b'$ each round in at least the $1-\alpha$ fraction of the rounds that the agent bids $0$. If $\alpha$ is small, the adversary's spending on the $\alpha$ fraction of the rounds that the agent has value is negligible. Hence, with a total budget of $(1-\alpha)T\approx T$, the adversary runs out of budget at time about $T\frac{1}{b'}$.
For the first $T\frac{1}{b'}$ rounds, the adversary still has budget remaining, and when the agent bids, she wins with probability $1-F(b')$.
For the remaining $T(1 - \frac{1}{b'})$ rounds, the adversary has no leftover budget, so the agent wins every time she bids.
Assuming the agent does not run out of budget (which we ensure later), the resulting fraction of rounds she wins out of the ones she bids in is
\begin{align} \label{eq:bound_using_cdf}
    \frac{1}{b'}\qty\big(1 - F(b'))
    +
    \qty(1 - \frac{1}{b'})
    =
    1 - \frac{F(b')}{b'}
    .
\end{align}

Under the above calculation, the agent wants to maximize the above quantity for any $b'$ by setting $F(b') = \lambda b'$, i.e., using a uniform distribution, and then optimizing over $\lambda$.
Specifically, $\lambda$ is maximized while ensuring the assumption we made above: the agent does not run out of budget for any bid $b'$ by the adversary.
This results in $\lambda = \frac{1}{1 + \sqrt 2}$, which is the uniform distribution we used above.

In the full proof in 
\ifarxiv
\cref{sec:app:lower},
\else
the \href{https://giannisfikioris.org/IJCAI25/paper.html}{full version} of this paper,
\fi
we have to account for the adversary using more complicated strategies, e.g., bidding as a function of time, remaining budget, etc.
We show that there is little benefit to such strategies, which is where the $\sqrt{\nicefrac{\log T}{T}}$ term comes from.
This is done by a careful martingale analysis of three key quantities: the agent's total spending, the agent's total utility, and the adversary's spending.
We show that with high probability, these three increase with the same rate under the adversary's optimal strategy.
This essentially makes the simplified analysis given above correct.

\section{Robust Strategy Upper Bounds} \label{sec:upper}

In this section, we provide upper bounds for the pseudo-market mechanism.
First, we show that our \RRB strategy for the $2 - \sqrt 2$ lower bound in \cref{sec:lower} is almost tight: for every fixed bidding distribution the agent might use, the adversary can bid to prevent the agent from achieving more than a $\frac35$ fraction of her ideal utility.
Second, we show that the adversary can bid so that the agent cannot guarantee more than a $1 - \frac 1 e$ fraction of her ideal utility under any strategy.

\subsection{Fixed bidding distribution \texorpdfstring{$\frac35$}{3/5} upper bound} \label{ssec:upper:fixed_distr}

In this section, we upper bound the robustness factor of an agent with $\text{Bernoulli}(\alpha)$ value distribution under the assumption that the agent must pick a fixed distribution to bid each round.
Specifically, the agent chooses a distribution $\mathcal D$ and bids $b[t] = r[t]V[t]$ at each time $t$ where the $r[t]$ are i.i.d. with distribution $\mathcal D$. We allow any distribution $\mathcal D$, even ones that depend on $T$; however, we assume the support of $\mathcal D$ is in some range $[0,\bar b]$ not depending on $T$.
Finally, as is standard with these upper bounds, we assume that $\alpha$ is small.
The final result is the following, showing that an agent cannot hope for better than $\frac 3 5$-robustness under such strategies.

\begin{restatable}{theorem}{fixeddistributionupperbound}
\label{thm:fixeddistributionupperbound}
Assume that an agent with fair share $\alpha$ and value distribution $\Bern(\alpha)$ bids according to distribution $\mathcal D$ whenever her value is $1$. Then this strategy cannot be $\beta$-robust for $\beta > \frac 3 5$ as $\alpha \to 0$.
\end{restatable}

The theorem's main lemma is the formal statement of our argument in \cref{ssec:lower:strategy}.
Specifically, we prove that if the agent's bidding strategy is $\beta$-robust, then its CDF must be upper-bounded by the CDF of a certain uniform distribution.
This implies that the agent must make high bids with decent probability.
Next, we will leverage this to show that if the agent bids too high, she will run out of budget.

\begin{restatable}{lemma}{cdfrestriction}
\label{lem:cdfrestriction}
Let $F(\cdot)$ be the CDF of a bidding distribution that is $\beta$-robust, as described in \cref{thm:fixeddistributionupperbound}. Then for $x \geq 1-\alpha$,
\begin{equation*}
    F(x-) \leq \frac{1-\beta}{1-\alpha}x
\end{equation*}
where $F(x-) = \lim_{y \to x-} F(y)$.
\end{restatable}

We prove the lemma by considering the adversary bidding $x$ every round, similar to our argument in \eqref{eq:bound_using_cdf}.
Specifically, if $x \ge 1$ (consider the case when $\alpha \to 0$), we know that the adversary will run out of budget approximately on round $\frac{T}{x}$.
While the adversary is bidding $x$, the agent gets a $1 - F(x-)$ fraction of the rounds she requests.
This means that, even if she always wins the item after the adversary runs out of budget, she wins at most
\begin{equation*}
    \frac{T}{x}\qty\big(1 - F(x-)) \alpha 
    +
    \qty(T - \frac{T}{x}) \alpha 
\end{equation*}
rounds where she has a value of $1$.
By robustness, this has to be at least $\beta T \alpha$, yielding the lemma.
The full proof, in which we consider the interactions of the agent and the adversary more carefully, is deferred to
\ifarxiv
\cref{sec:app_upper}.
\else
the \href{https://giannisfikioris.org/IJCAI25/paper.html}{full version} of this paper.
\fi

The final step for proving \cref{thm:fixeddistributionupperbound} is to show that if $\beta > \frac{3}{5}$, then the agent's bidding is too high.
Specifically, if the adversary uses a constant bid $b'$ every round then the agent has to pay in expectation
\begin{align*}
    \E_{b \sim \mathcal D}[b \mid b \ge b']
\end{align*}
every time she wins the resource.
Using \cref{lem:cdfrestriction}, we can show that if $\beta > \frac{3}{5}$, then under the right $b'$, the above payment is too high, and the agent will run out of money before accumulating a $\beta$ fraction of her ideal utility.

We defer the complete proof of \cref{thm:fixeddistributionupperbound} to
\ifarxiv
\cref{sec:app_upper}.
\else
the \href{https://giannisfikioris.org/IJCAI25/paper.html}{full version} of this paper.
\fi

\subsection{Constructive proof for an upper bound of \texorpdfstring{$1-\frac1e + \frac\alpha e + o(1)$}{1 - 1/e}}
\label{ssec:upper:arbitrary_distr}

\cite{fikioris2023online} prove that no mechanism can guarantee every agent a $1 - \frac{1}{e}$ fraction of her ideal utility when agents have Bernoulli valuations.
While this bound remains the best-known upper bound for robust guarantees, its proof is existential.
In particular, this result does not provide insight into how the adversary can prevent an agent from obtaining more than such a fraction of her ideal utility.
In this section, we provide a simple bidding strategy that the adversary can follow to ensure that the agent does not get more than a $1 - \frac{1}{e}$ fraction of her ideal utility.
This result is stronger than the one in \cref{thm:fixeddistributionupperbound} in two ways.
First, the adversary's strategy is independent of the agent's strategy.
Second, the upper bound holds for any strategy the agent might follow, even time-varying ones.
Specifically, we prove the following theorem.

\begin{restatable}{theorem}{arbitrarystrategyupperbound}
\label{thm:arbitrarystrategyupperbound}
There exists a stationary bidding policy by the adversary such that an agent with fair share $\alpha$ and value distribution $\Bern(\alpha)$ cannot get expected utility more than a $1 - \frac{1}{e} + \frac\alpha e + O(\sqrt{\nicefrac{\log T}{T}})$ fraction of her ideal utility, i.e.,
\begin{equation*}
    \sum_{t=1}^T \mathbb E[U[t]] \leq \alpha T\left(1-\frac1e + \frac\alpha e\right) + O(\sqrt{T\log T}).
\end{equation*}
\end{restatable}

The bidding policy we use for the theorem's proof is roughly (for simplicity we assume the adversary's budget constraint holds only in expectation) the following CDF (depicted also in \cref{fig:adversary_cdf})
\begin{equation*}
    F'(b') =
    \begin{cases}
        0 & \text{if $b' < 0$}\\
        \frac{1 - \alpha}{e-(e-2)b'} + \alpha & \text{if $0\leq b' \leq \frac{e-1}{e-2}$}\\
        1 & \text{if $b' > \frac{e-1}{e-2}$}
    \end{cases}.
\end{equation*}

\begin{figure}
\centering
\includegraphics[width=.6\linewidth]{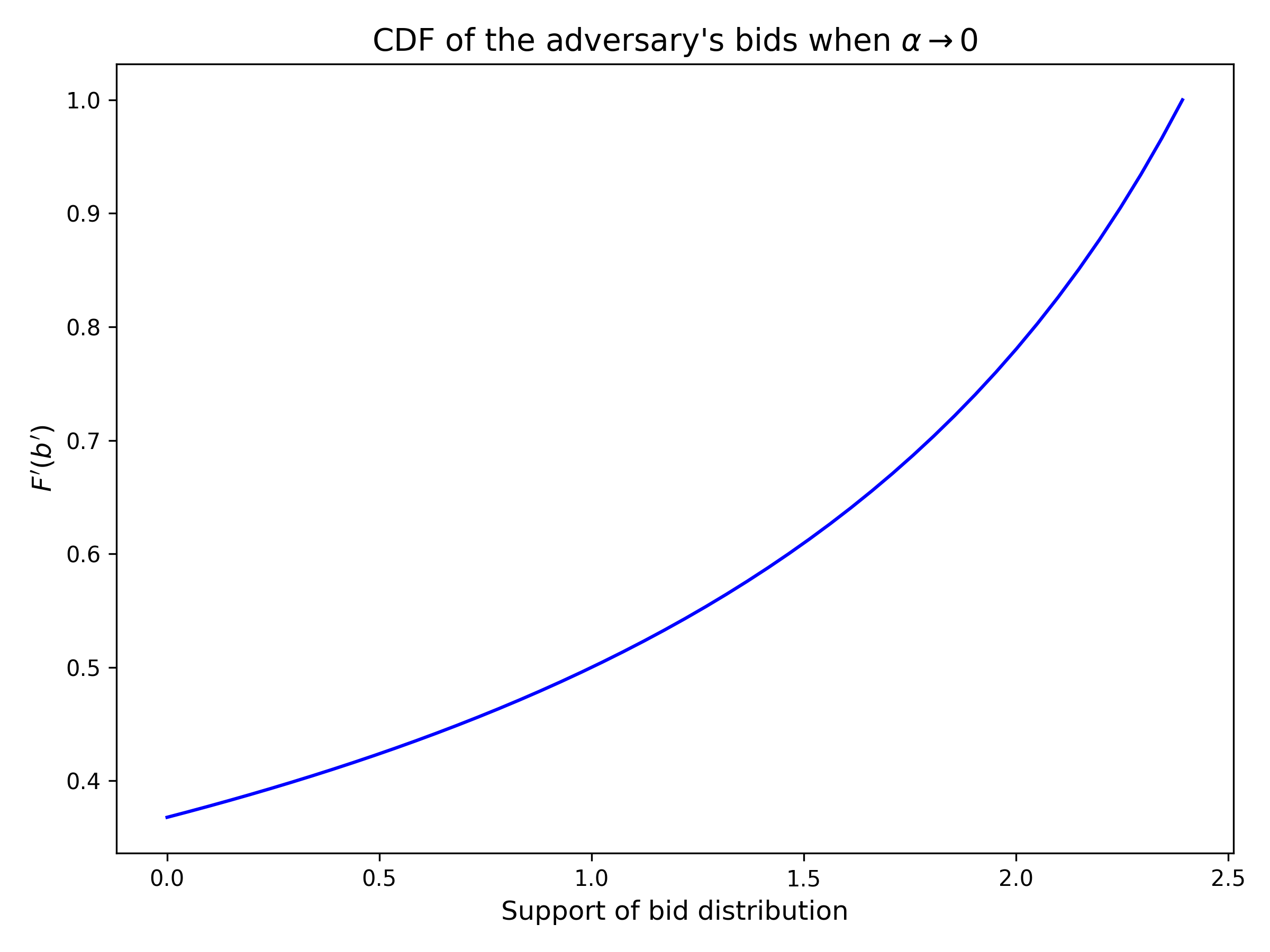}
\caption{CDF of the adversary's bid distribution used in \cref{thm:arbitrarystrategyupperbound} when $\alpha\to 0$. The expected bid under this distribution is $1$, so the adversary will not run out of budget. The CDF is carefully chosen such that the agent cannot win more than $1-1/e$ fraction of the rounds regardless of strategy.}
\label{fig:adversary_cdf}
\end{figure}

A first observation from the above distribution is that the adversary's expected bid is $1 - \alpha$.
This implies that the adversary will not run out of budget in expectation, even if the agent never wins any rounds.
The next key observation is that the highest bid the agent can use to not run out of budget in expectation is $b \overset{\alpha \to 0}{=} \frac{e}{e - 1}$.
In particular, this bid wins with probability $F'(b) = \frac{e - 1 + \alpha}{e}$, which is the claimed bound.

For the full proof that can be found in
\ifarxiv
\cref{sec:app_upper},
\else
the \href{https://giannisfikioris.org/IJCAI25/paper.html}{full version} of this paper,
\fi
we present a more careful analysis of what happens when the agent follows \textit{any} bidding strategy.
In particular, we show that with high probability, the adversary will not run out of budget.
Again with high probability, we show that any sequence of bids by the agent, even if they change according to the past, cannot guarantee more than the claimed fraction of ideal utility.

\section{Experimental Evaluation of \texorpdfstring{\RRB}{Randomized Robust Bidding}}
\label{sec:experiments}

In our theoretical results, we investigated the worst-case utility guarantees that we could obtain under arbitrary behavior by the other agents, which includes adversarial (and collusive) behavior that may not be realistic.
In this section, we experimentally investigate the fraction of ideal utility an agent gets when all agents use robust strategies and show that our proposed strategy performs very well.
Specifically, we compare the agents' utilities under the following two strategies. First, all agents use the deterministic $(1/2-o(1))$-robust strategy given by \cite{gorokh2019remarkable}, where each agent bids $2$ each time their value is in the top $\alpha_i$-quantile of their value distribution. Second, all agents use our \strategy strategy, where each agent bids according to a uniform distribution instead.
To also illustrate our theoretical results where the other agents are behaving adversarially, we also run an experiment where one agent is using \strategy but the other agents adversarially always bid $1$ regardless of their values.

\begin{figure}[t!]
\centering
\includegraphics[width=.6\linewidth]{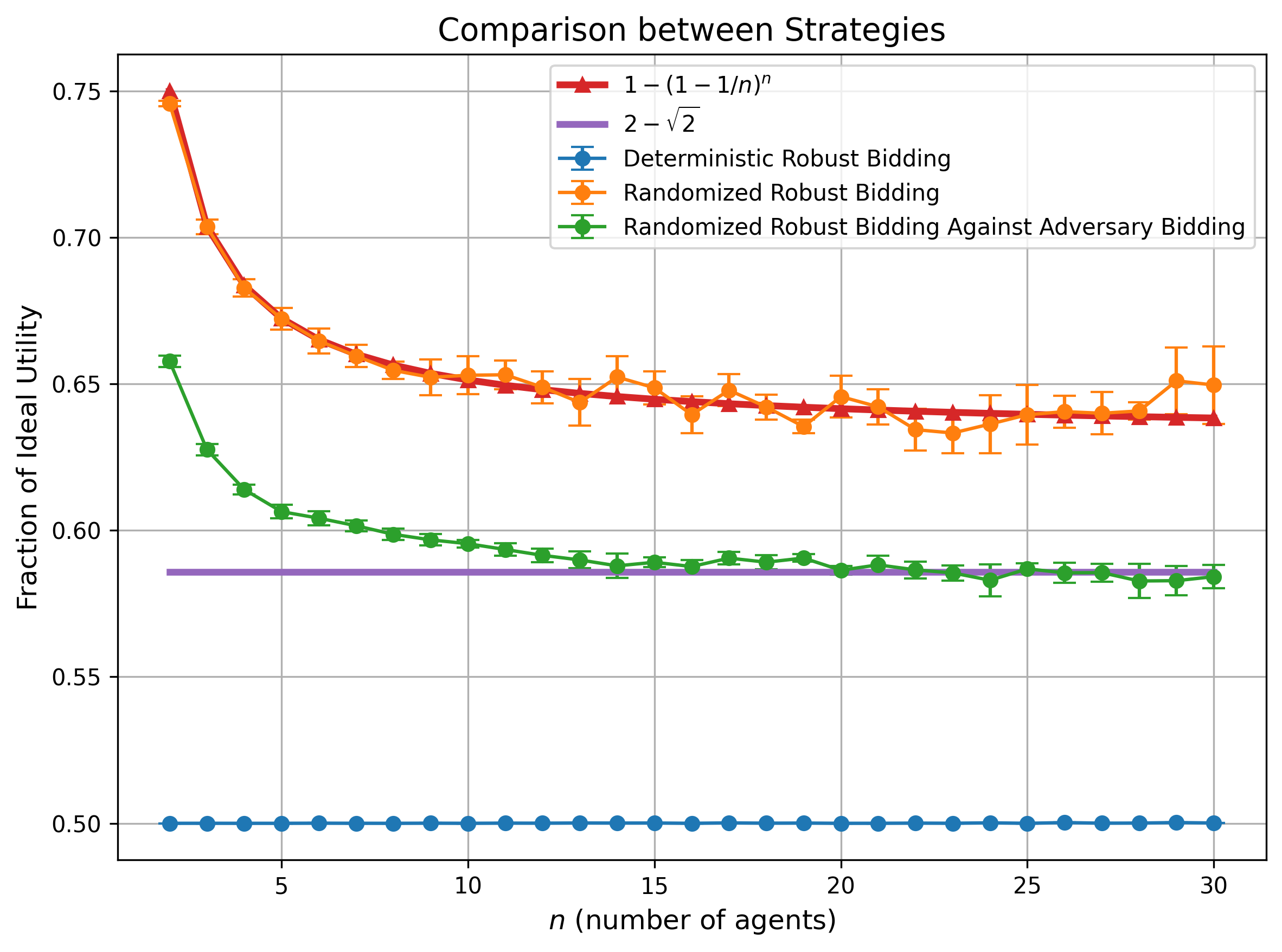}
\caption{Fraction of ideal utility that an agent obtains under differing strategy profiles. We compare the agents' utility when they all use the previously best-known robust strategy from \protect\cite{gorokh2019remarkable}, labeled Deterministic Robust Bidding, with the agents' utility when they all use \strategy. 
We also plot an agent's utility when they use \strategy but the other agents adversarially always bid $1$ regardless of their values, labeled Randomized Robust Bidding against Adversary.
When all agents use \strategy, they achieve $\approx1 - (1-1/n)^n$ fraction of their ideal utility, the theoretical maximum for any allocation procedure.
When one agent uses \strategy but the other agents behave adversarially, the agent using \strategy achieves at least a $2-\sqrt2$ fraction, the guarantee of \cref{thm:lowerBound}.}
\label{fig:comparison-experiment}
\end{figure}

We consider the symmetric agent case, where each agent has fair share $\alpha_i=1/n$.
We consider each agent's value distribution to be $\Bern(1/n)$.
For each strategy, we compare the agents' resulting utility for each number of players $n\in\{2,3,\dots,30\}$.
We ran the mechanism for $T=100000$ time periods $10$ times and recorded the average fraction of ideal utility that a particular agent obtained\footnote{Our code can be found at \url{https://github.com/davidxlin/repeated-fisher-market-experiments}}.
We plot our results in \cref{fig:comparison-experiment}.

We can see that when every agent uses the deterministic strategy, each agent gets $1/2$ of their ideal utility, similar to the theoretical guarantee.
When each player plays our randomized strategy, they enjoy a higher utility. In particular, they achieve close to a $1 - (1 - 1/n)^n$ fraction of their ideal utility.
This is the best we can hope for: no allocation procedure can guarantee each agent a greater fraction of their ideal utility when they have the aforementioned Bernoulli values.
This is superior to our theoretical lower bound of $2-\sqrt2\approx 0.59$ in \cref{thm:lowerBound} and shows the empirical performance of \RRB can be even greater than the one under worst-case competition.

\clearpage
\bibliographystyle{plain}
\bibliography{ijcai25}

\clearpage
\appendix
\section{Deferred proofs of Section \ref{sec:lower}}
\label{sec:app:lower}

In this section, we present the full proof of \cref{thm:lowerBound}, the $(2-\sqrt 2)$-robust strategy lower bound.
We first restate the theorem for completeness.

\lowerBound*

\begin{proof}[Proof of \cref{thm:lowerBound}]
By \cref{lem:bernoullivaluedistributionisworstcase}, we can assume without loss of generality that $V[t]\sim\Bern(\alpha)$. Using our reduction in the proof of \cref{lem:bernoullivaluedistributionisworstcase}, the \strategy strategy reduces to just bidding $b[t] = r[t]V[t]\pmb1\{t\leq\tau\}$, where we recall that $\tau = \max\{t\geq 1: B[t] \geq \bar b\}$.

We assume without loss of generality that each adversary bid $b'[t]\leq \bar b$. Let $\mathcal H_t$ denote the history up to and including time $t$. Let $\mathcal G_t$ be the $\sigma$-algebra generated by $\mathcal H_t$ and $b'[t+1]$. Define 
\begin{align*}
    M_1[t] & = \sum_{s=1}^{\min\{t, \tau\}} P[s] - \sum_{s=1}^{\min\{t, \tau\}} \mathbb E[P[s]\mid\mathcal G_{s-1}],\\
    M_2[t] & = \sum_{s=1}^{\min\{t, \tau\}} U[s] - \sum_{s=1}^{\min\{t, \tau\}} \mathbb E[U[s]\mid\mathcal G_{s-1}],\\
    M_3[t] & = \sum_{s=1}^t b'[s]V[s] - \sum_{s=1}^t \mathbb E[b'[s]V[s]\mid\mathcal G_{s-1}].
\end{align*}
Observing that $\tau$ is a stopping time with respect to the filtration $\mathcal G_t$, by definition, $M_1[t]$, $M_2[t]$, and $M_3[t]$ are martingales with respect to $\mathcal G_t$. Let $\epsilon>0$ and define the event
\begin{equation*}
     E = \{M_1[T] < \epsilon\bar b, M_2[T] > -\epsilon, M_3[T] < \epsilon\bar b\}.
\end{equation*}
Observe that the increments of $M_1[t]$ and $M_3[t]$ are almost surely bounded by $\bar b$ and the increments of $M_2[t]$ are almost surely bounded by $1$, so by the Azuma-Hoeffding inequality,
\begin{equation}
\label{eq:PrE}
    \Pr(E) \geq 1 - 3\exp\left(-\frac{\epsilon^2}{2\bar bT}\right).
\end{equation}
In what follows, consider what happens on the high probability event $E$.

Observe that $b'[t]$ is $\mathcal G_{t-1}$-measurable and $V[t]$ is independent of $\mathcal G_{t-1}$, so
\begin{equation*}
    \mathbb E[b'[t]V[t]\mid\mathcal G_{t-1}] = \mathbb E[V[t]] = \alpha b'[t].
\end{equation*}
Then, by the definition of $E$,
\begin{equation*}
\begin{split}
    \sum_{t=1}^T b'[t]V[t] & = \sum_{t=1}^T \mathbb E[b'[t]V[t]\mid\mathcal G_{t-1}] + M_3[T]\\
    &\leq \alpha \sum_{t=1}^T b'[t] + \epsilon\bar b. 
\end{split}
\end{equation*}
This implies
\begin{equation*}
\begin{split}
    \sum_{t=1}^T b'[t] & = \sum_{t=1}^T b'[t](1-V[t]) + \sum_{t=1}^T b'[t]V[t]\\
    & \leq \sum_{t=1}^T b'[t](1-V[t]) + \alpha \sum_{t=1}^T b'[t] + \epsilon\bar b\\
    & \leq (1-\alpha)T + \alpha \sum_{t=1}^T b'[t] + \epsilon\bar b
\end{split}
\end{equation*}
where the last inequality comes from the adversary's budget constraint of $(1-\alpha)T$, observing that if $V[t]=0$, then adversary always wins and pays their bid. Rearranging,
\begin{equation}
\label{eq:adversarybiddingonE}
    \sum_{t=1}^T b'[t] \leq T + \frac{\epsilon\bar b}{1-\alpha}.
\end{equation}

Now we compute the expected payments $P[t]$ and utilities $U[t]$ conditioned on $\mathcal G_{t-1}$ on the event $\{t\leq \tau\}$, using the fact that $b'[t]$ is $\mathcal G_{t-1}$-measurable and $b[t]$ is independent of $\mathcal G_{t-1}$:
\begin{equation*}
\begin{split}
    \mathbb E[P[t]\mid\mathcal G_{t-1}] & =\mathbb E[b[t]\pmb1\{b[t]>b'[t]\}\mid\mathcal G_{t-1}]\\
    & =\mathbb E[r[t]V[t]\pmb1\{r[t]>b'[t]\}\mid\mathcal G_{t-1}]\\
    & = \alpha\int_{b'[t]}^{\bar b}\frac{x}{\bar b}\,dx = \alpha\left(\frac{\bar b^2 - b'[t]^2}{2\bar b}\right)
\end{split}
\end{equation*}
and
\begin{equation*}
\begin{split}
   \mathbb E[U[t]\mid\mathcal G_{t-1}] & = \Pr(b[t]>b'[t]\mid\mathcal G_{t-1})\\
   & = \Pr(r[t]V[t]>b'[t]\mid\mathcal G_{t-1})\\
   & = \alpha\int_{b'[t]}^{\bar b}\frac{dx}{\bar b} = \alpha\left(1 - \frac{b'[t]}{\bar b}\right).
\end{split}
\end{equation*}

The agent's utility can then be bounded as
\begin{equation}
\begin{split}
\label{eq:utilityonE}
    \sum_{t=1}^T U[t] & = \sum_{t=1}^{\min\{T,\tau\}} U[t]\\
    & = \sum_{t=1}^{\min\{T,\tau\}}\mathbb E[U[t]\mid\mathcal G_{t-1}] + M_2[T]\\
     &\geq \sum_{t=1}^{\min\{T,\tau\}}\alpha\left(1-\frac{b'[t]}{\bar b}\right) - \epsilon\\
     &= \alpha\left(\min\{T,\tau\}-\frac{1}{\bar b}\sum_{t=1}^{\min\{T,\tau\}} b'[t]\right)-\epsilon,
\end{split}
\end{equation}
which we shall lower bound in cases. First, consider what happens on $E\cap\{\tau = T\}$. We use \eqref{eq:adversarybiddingonE} to obtain
\begin{equation*}
\begin{split}
    \alpha\left(T-\frac{1}{\bar b}\sum_{t=1}^T b'[t]\right) & \geq \alpha\left(T - \frac{1}{\bar b}\left(T + \frac{\epsilon\bar b}{1-\alpha}\right)\right)\\
    & = \alpha T\left(1-\frac{1}{\bar b}\right) - \frac{\alpha}{1-\alpha}\epsilon.
\end{split}
\end{equation*}
Using \eqref{eq:utilityonE}, we obtain
\begin{equation}
\label{eq:utilityonEwhennotrunningoutofmoney}
    \frac{1}{\alpha T}\sum_{t=1}^T U[t] \geq \left(1-\frac{1}{\bar b}\right) - \frac{\epsilon}{(1-\alpha)T} - \frac{\epsilon}{T}.
\end{equation}

Otherwise, consider what happens on $E\cap\{\tau < T\}$. If $\tau < T$, by definition of $\tau$,
\begin{equation}
\label{eq:paymentishighwhenrunningoutofmoney}
    \sum_{t=1}^\tau P[t] \geq \alpha T - \bar b.
\end{equation}
By the definition of $E$,
\begin{equation}
\label{eq:paymentonE}
    \alpha\left(\frac{\bar b\tau}{2} - \frac{1}{2\bar b}\sum_{t=1}^\tau b'[t]^2\right) = \sum_{t=1}^\tau \mathbb E[P[t]\mid\mathcal G_{t-1} = \sum_{t=1}^\tau P[t] - M_1[T] > \sum_{t=1}^\tau P[t] - \bar b\epsilon.
\end{equation}
Combining \eqref{eq:paymentishighwhenrunningoutofmoney} and \eqref{eq:paymentonE} yields
\begin{equation*}
    \alpha\left(\frac{\bar b\tau}{2} - \frac{1}{2\bar b}\sum_{t=1}^\tau b'[t]^2\right)\geq  \alpha T - \bar b - \bar b\epsilon,
\end{equation*}
which implies
\begin{equation}
\label{eq:boundsumsquaresadversarybids}
    \sum_{t=1}^\tau b'[t]^2 \leq 2\bar b\left(\frac{\bar b\tau }{2} - T\right) + \frac{2\bar b + 2\bar b\epsilon}{\alpha}.
\end{equation}
We can use the obvious bound $\sum_{t=1}^\tau b'[t]^2\geq 0$ with \eqref{eq:boundsumsquaresadversarybids} to obtain
\begin{equation}
\label{eq:lowerboundstoppingtime}
    \tau \geq \frac{2T}{\bar b} - \frac{2\bar b + 2\bar b\epsilon}{\alpha \bar b}.
\end{equation}
By the Cauchy-Schwartz inequality, we have
\begin{equation}
\label{eq:sumbidsbidssquaredinequality}
    \sum_{t=1}^\tau b'[t] \leq \sqrt{\tau\sum_{t=1}^\tau b'[t]^2}.
\end{equation}
Using \eqref{eq:utilityonE} followed by \eqref{eq:boundsumsquaresadversarybids} combined with \eqref{eq:sumbidsbidssquaredinequality},
\ifarxiv 
\begin{equation}
\label{eq:utilityboundonEwhenrunoutofmoney}
\begin{split}
    \frac{1}{\alpha T}\sum_{t=1}^T U[t] & \geq \frac{\tau}{T}-\frac{1}{\bar b T}\sum_{t=1}^\tau b'[t]-\frac{\epsilon}{\alpha T}\\
    &=\frac{\tau}{T} - \frac{1}{\bar bT}\sqrt{\tau\left(2\bar b\left(\frac{\bar b\tau}{2}-T\right)\right) + \frac{2\bar b + 2\bar b\epsilon}{\alpha}} - \frac{\epsilon}{\alpha T}.
\end{split}
\end{equation}
\else
\begin{equation}
\label{eq:utilityboundonEwhenrunoutofmoney}
\begin{split}
    &\frac{1}{\alpha T}\sum_{t=1}^T U[t] \geq \frac{\tau}{T}-\frac{1}{\bar b T}\sum_{t=1}^\tau b'[t]-\frac{\epsilon}{\alpha T}\\
    & =\frac{\tau}{T} - \frac{1}{\bar bT}\sqrt{\tau\left(2\bar b\left(\frac{\bar b\tau}{2}-T\right)\right) + \frac{2\bar b + 2\bar b\epsilon}{\alpha}}\\
    &\quad - \frac{\epsilon}{\alpha T}.
\end{split}
\end{equation}
\fi 
We bound the above in cases. If $\tau < 2T/\bar b$, then using \eqref{eq:lowerboundstoppingtime}, \eqref{eq:utilityboundonEwhenrunoutofmoney} is at least 
\begin{equation}
\label{eq:utilityboundonEwhenrunoutofmoneykindafast}
\begin{split}
    \frac{2}{\bar b} - \frac{2\bar b + 2\bar b\epsilon}{\alpha \bar b T} - \frac{1}{\bar b T}\sqrt{\frac{2\bar b + 2\bar b\epsilon}{\alpha}} - \frac{\epsilon}{\alpha T}.
\end{split}
\end{equation}
Otherwise, if $\tau \geq 2T/\bar b$, then \eqref{eq:utilityboundonEwhenrunoutofmoney} is at least
\ifarxiv 
\begin{equation}
\label{eq:utilityboundonEwhenrunoutofmoneynotsuperfast}
\begin{split}
    & \frac{\tau}{T} - \frac{1}{\bar b}\sqrt{\frac{\tau}{T}\left(2\bar b\left(\frac{\bar b\tau}{2T} - 1\right)\right)} - \frac{1}{\bar b T}\sqrt{\frac{2\bar b + 2\bar b\epsilon}{\alpha}} - \frac{\epsilon}{\alpha T}\\
    & \quad \geq \inf_{x\in [2/\bar b, 1)}\left[x - \frac{1}{\bar b}\sqrt{x\left(2\bar b\left(\frac{\bar bx}{2}-1\right)\right)}\right] -\frac{1}{\bar b T}\sqrt{\frac{2\bar b + 2\bar b\epsilon}{\alpha}} - \frac{\epsilon}{\alpha T}\\
    & \quad = 1 - \frac{1}{\bar b}\sqrt{\bar b^2-2\bar b}- \frac{1}{\bar b T}\sqrt{\frac{2\bar b + 2\bar b\epsilon}{\alpha}} - \frac{\epsilon}{\alpha T}.
\end{split}
\end{equation}
\else
\begin{equation}
\label{eq:utilityboundonEwhenrunoutofmoneynotsuperfast}
\begin{split}
    \frac{\tau}{T}& - \frac{1}{\bar b}\sqrt{\frac{\tau}{T}\left(2\bar b\left(\frac{\bar b\tau}{2T} - 1\right)\right)}\\
    & - \frac{1}{\bar b T}\sqrt{\frac{2\bar b + 2\bar b\epsilon}{\alpha}} - \frac{\epsilon}{\alpha T}\\
    \geq &\inf_{x\in [2/\bar b, 1)}\left[x - \frac{1}{\bar b}\sqrt{x\left(2\bar b\left(\frac{\bar bx}{2}-1\right)\right)}\right]\\
    &-\frac{1}{\bar b T}\sqrt{\frac{2\bar b + 2\bar b\epsilon}{\alpha}} - \frac{\epsilon}{\alpha T}\\
    =& 1 - \frac{1}{\bar b}\sqrt{\bar b^2-2\bar b}- \frac{1}{\bar b T}\sqrt{\frac{2\bar b + 2\bar b\epsilon}{\alpha}} - \frac{\epsilon}{\alpha T}.
\end{split}
\end{equation}
\fi
By combining \eqref{eq:PrE}, \eqref{eq:utilityonEwhennotrunningoutofmoney}, \eqref{eq:utilityboundonEwhenrunoutofmoneykindafast}, \eqref{eq:utilityboundonEwhenrunoutofmoneynotsuperfast}, and Markov's inequality, we obtain
\ifarxiv 
\begin{equation*}
\begin{split}
    \frac{1}{\alpha T}\sum_{t=1}^T \mathbb E[U[t]] \geq & \left(1 - 3\exp\left(-\frac{\epsilon^2}{2\bar b T}\right)\right)\min\Bigg\{\left.\left(1-\frac{1}{\bar b}\right) - \frac{\epsilon}{(1-\alpha)T} - \frac{\epsilon}{T},\right.\\
    & \left.\min\left\{\frac{2}{\bar b} - \frac{2\bar b + 2\bar b\epsilon}{\alpha \bar b T},\right.\right. \left. 1 - \frac{1}{\bar b}\sqrt{\bar b^2-2\bar b} \right\}- \frac{1}{\bar b T}\sqrt{\frac{2\bar b + 2\bar b\epsilon}{\alpha}} - \frac{\epsilon}{\alpha T}\Bigg\}.
\end{split}
\end{equation*}
\else
\begin{equation*}
\begin{split}
    & \frac{1}{\alpha T}\sum_{t=1}^T \mathbb E[U[t]] \geq \left(1 - 3\exp\left(-\frac{\epsilon^2}{2\bar b T}\right)\right)\min\Bigg\{\\
    & \left.\left(1-\frac{1}{\bar b}\right) - \frac{\epsilon}{(1-\alpha)T} - \frac{\epsilon}{T},\min\left\{\frac{2}{\bar b} - \frac{2\bar b + 2\bar b\epsilon}{\alpha \bar b T},\right.\right.\\
    & \left. 1 - \frac{1}{\bar b}\sqrt{\bar b^2-2\bar b} \right\}- \frac{1}{\bar b T}\sqrt{\frac{2\bar b + 2\bar b\epsilon}{\alpha}} - \frac{\epsilon}{\alpha T}\Bigg\}.
\end{split}
\end{equation*}
\fi
By setting $\epsilon = \sqrt{T\ln T}$ and substituting $\bar b = 1+\sqrt{2}$, we obtain the theorem statement.
\end{proof}
\section{Deferred proofs of Section \ref{sec:upper}}
\label{sec:app_upper}

\subsection{Proof of Theorem \ref{thm:fixeddistributionupperbound}}

In this section, we provide the full proof of \cref{thm:fixeddistributionupperbound}, which we first restate for completeness.

\fixeddistributionupperbound*

We assume without loss of generality that the adversary's bids are also uniformly bounded by $\bar b$, the upper bound of the support of the agent's bid distribution. For expository convenience, we assume the agent will win in case of ties when $b'[t] = b[t]$. Because we are trying to give an upper bound on the agent's utility, we make the following assumptions without loss of generality. The agent can bid anything until she has negative tokens instead of more strictly enforcing the budget constraint. Contrarily, the adversary will not be allowed to bid unless their budget is at least $\bar b$. 

With these assumptions, let us write the agent's utility more explicitly. Let $\tau = \max\{t\leq T:B[t] \geq 0\}$ be the last time that the agent has tokens, and let $\tau' = \max\{t\leq T: B'[t] \geq \bar b\}$ be the last time the adversary is allowed to bid. The agent's utility at time $t$ can be written as
\begin{equation*}
    U[t] = \begin{cases}V[t] & \text{if $t\leq \tau$, and either $t > \tau'$ or $r[t]\geq b'[t]$}\\0 & \text{otherwise}\end{cases}.
\end{equation*}
We can then write the agent's objective as
\begin{equation}
\label{eq:agentobjective}
    \sum_{t=1}^T U[t] = \sum_{t=1}^{\tau\wedge\tau'}V[t]\pmb1\{r[t]\geq b'[t]\} + \sum_{t=\tau'+1}^\tau V[t],
\end{equation}
where $a\wedge b$ denotes the minimum of $a$ and $b$.

Suppose the adversary has the strategy of bidding every round from a fixed distribution $\mathcal D'$ until their budget goes below $\bar b$. In this case, the agent's expected utility can be more easily analyzed. Let $r\sim \mathcal D$ and $b'\sim\mathcal D'$. Define
\ifarxiv 
\begin{equation}
\begin{split}
\label{eq:agentmodifiedobjective}
    U & (\alpha,\mathcal D, \mathcal D') = \alpha T\left(\rule{0cm}{1.0cm}\right.\Pr(r \geq b')\min\left\{1, \frac{1}{\mathbb E[r\pmb1\{r\geq b'\}]},\right. \left.\left.\frac{1-\alpha}{(1-\alpha)\mathbb E[b'] + \alpha\mathbb E[b'\pmb1\{b' > r\}]}\right\}\right.\\
    & + \left(\rule{0cm}{0.8cm}\right.\min\Bigg\{1 \left.\left.- \frac{1-\alpha}{(1-\alpha)\mathbb E[b'] + \alpha\mathbb E[b'\pmb1\{b' > r\}]},\right.\right.\left.\left.\frac{1 - \frac{1-\alpha}{(1-\alpha)\mathbb E[b'] + \alpha\mathbb E[b'\pmb1\{b' > r\}]}\mathbb E[r\pmb1\{r\geq b'\}]}{\mathbb E[r]}\right.\right.\Bigg\}\left.\rule{0cm}{0.8cm}\right)^+\left.\rule{0cm}{1.0cm}\right).
\end{split}
\end{equation}
\else
\begin{equation}
\begin{split}
\label{eq:agentmodifiedobjective}
    U & (\alpha,\mathcal D, \mathcal D') = \alpha T\left(\rule{0cm}{1.0cm}\right.\Pr(r \geq b')\\
    & \quad\min\left\{1, \frac{1}{\mathbb E[r\pmb1\{r\geq b'\}]},\right.\\
    & \left.\left.\quad \frac{1-\alpha}{(1-\alpha)\mathbb E[b'] + \alpha\mathbb E[b'\pmb1\{b' > r\}]}\right\}\right.\\
    & \quad + \left(\rule{0cm}{0.8cm}\right.\min\Bigg\{1\\
    & \quad \left.\left.- \frac{1-\alpha}{(1-\alpha)\mathbb E[b'] + \alpha\mathbb E[b'\pmb1\{b' > r\}]},\right.\right.\\
    & \left.\left.\quad\frac{1 - \frac{1-\alpha}{(1-\alpha)\mathbb E[b'] + \alpha\mathbb E[b'\pmb1\{b' > r\}]}\mathbb E[r\pmb1\{r\geq b'\}]}{\mathbb E[r]}\right.\right.\\
    & \quad\Bigg\}\left.\rule{0cm}{0.8cm}\right)^+\left.\rule{0cm}{1.0cm}\right).
\end{split}
\end{equation}
\fi
We claim that $U(\alpha,\mathcal D,\mathcal D')$ is roughly the agent's expected utility. The formula can be derived from the following intuition. When both the agent and adversary have tokens left, the agent will win with probability $\alpha\Pr(r\geq b')$, the agent's expected payment is $\alpha\mathbb E[r\pmb1\{r\geq b'\}]$, and the adversary's expected payment is $\frac{1-\alpha}{1-\alpha\mathbb E[b'] + \alpha\mathbb E[b'\pmb1\{b'>r\}]}$. If the agent's expected payment is greater than the adversary's expected payment, then the agent will run out of tokens at time about $\min\left\{T, \frac{T}{\mathbb E[r\pmb1\{r\geq b'\}]}\right\}$, which is before the adversary runs out of tokens, giving her a total expected utility of about $\alpha T\min\left\{1, \frac{1}{\mathbb E[r\pmb1\{r\geq b'\}]}\right\}$. Otherwise, if the agent's expected payment is less than the adversary's expected payment, then the adversary runs out of tokens first at time about $\frac{(1-\alpha)T}{(1-\alpha)\mathbb E[b'] + \alpha\mathbb E[b'\pmb1\{b'>r\}]}$. Up to this time, the agent has won $\alpha T\Pr(r\geq b')$ fraction of the items. She spent about $\alpha \mathbb E[r\pmb1\{r\geq b'\}]$ tokens per round. There are only so many rounds left, of which only $\alpha$ fraction of those times $t$ the agent has $V[t]=1$, so she can only get at most $\alpha T\left(1-\frac{(1-\alpha)}{(1-\alpha)\mathbb E[b'] + \alpha\mathbb E[b'\pmb1\{b'>r\}]}\right)$ ideal utility more. The agent also has a budget constraint. After the adversary runs out of tokens, she has about $\alpha(\mathbb E[r\pmb1\{r\geq b'\}])$ tokens left and spends around $\mathbb E[r]$ tokens each time she requests, so she also has the additional bound of $\alpha T\left(\frac{1 - \frac{1-\alpha}{(1-\alpha)\mathbb E[b'] + \alpha\mathbb E[b'\pmb1\{b' > r\}]}\mathbb E[r\pmb1\{r\geq b'\}]}{\mathbb E[r]}\right)$ on her ideal utility gained after the adversary runs out of tokens. This intuition gives us the formula in \eqref{eq:agentmodifiedobjective}, which is formalized in the below lemma.
\begin{lemma}
\label{lem:agentmodifiedobjective}
The expectation of \eqref{eq:agentobjective} is at most $O(\sqrt{T\log T})$ more than \eqref{eq:agentmodifiedobjective}.
\end{lemma}
\begin{proof}
Let $X[1], \dots, X[t]$ be i.i.d. drawn from $\mathcal D$ with probability $\alpha$ and $0$ with probability $1-\alpha$ and $X'[1], \dots, X'[t]$ be i.i.d. drawn from $\mathcal D'$ such that $b[t] = X[t]\pmb1\{t\leq \tau\}$ and $b'[t] = X'[t]\pmb1\{t\leq\tau'\}$. Define
\begin{align*}
    & t_1 = \min\left\{T, \left\lfloor\frac{\alpha T + \bar b\sqrt{T\ln T}}{\alpha\mathbb E[r\pmb1\{r\geq b'\}]}\right\rfloor\right\},\\
    & t_2 = \max\left\{1, \left\lceil\frac{(1-\alpha)T - \bar b - \bar b\sqrt{T\ln T}}{(1-\alpha)\mathbb E[b'] + \alpha\mathbb E[b'\pmb1\{b' > r\}]}\right\rceil\right\},\\
    & t_3 = \min\left\{T, \left\lfloor\frac{(1-\alpha)T - \bar b + \bar b\sqrt{T\ln T}}{(1-\alpha)\mathbb E[b'] + \alpha\mathbb E[b'\pmb1\{b' > r\}]}\right\rfloor\right\},\\
    & t_4 = \min\left\{T, \left\lfloor\frac{T - t_2\mathbb E[r\pmb1\{r\geq b'\}] + t_3\mathbb E[r]}{\mathbb E[r]}\right\rfloor\right\}.
\end{align*}
Let $E$ be the event where the following hold.
\ifarxiv 
\begin{equation}
\label{eq:spendingbeforet1}
\begin{split}
    \sum_{s=1}^{t_1} X[s]\pmb1\{X[s] \geq X'[s]\} > \alpha t_1\mathbb E[r\pmb1\{r\geq b'\}] - \bar b\sqrt{T\ln T}
\end{split}
\end{equation}
\begin{equation}
\label{eq:adversaryspendingbeforet2}
\begin{split}
    \sum_{s=1}^{t_2} X'[s]\pmb1\{X[s] < X'[s]\} < t_2\Big((1-\alpha)\mathbb E[b']
     + \alpha\mathbb E[b'\pmb1\{b' > r\}]\Big) +  \bar b\sqrt{T\ln T}
\end{split}
\end{equation}
\begin{equation}
\label{eq:utilitybeforet1wedget2}
\begin{split}
    \sum_{s=1}^{t_1\wedge t_2}V[s]\pmb1\{X[s] \geq X'[s]\} < \alpha\Pr(r \geq b')(t_1\wedge t_2)+ \sqrt{T\ln T}
\end{split}
\end{equation}
\begin{equation}
\label{eq:spendingbeforet2}
\begin{split}
    \sum_{s=1}^{t_2}X[s]\pmb1\{X[s]\geq X'[s]\} > \alpha t_2\mathbb E[r\pmb1\{r\geq b'\}]- \bar b\sqrt{T\ln T}
\end{split}
\end{equation}
\begin{equation}
\label{eq:adversaryspendingbeforet3}
\begin{split}
    \sum_{s=1}^{t_3}X'[s]\pmb1\{X[s]<X'[s]\} > t_3\Big((1-\alpha)\mathbb E[b'] + \alpha\mathbb E[b'\pmb1\{b'<r\}]\Big) -  \bar b\sqrt{T\ln T}
\end{split}
\end{equation}
\else
\begin{equation}
\label{eq:spendingbeforet1}
\begin{split}
    \sum_{s=1}^{t_1} X[s]\pmb1\{X[s] \geq X'[s]\} > \alpha t_1\mathbb E[r\pmb1\{r\geq b'\}]\\
     - \bar b\sqrt{T\ln T}
\end{split}
\end{equation}
\begin{equation}
\label{eq:adversaryspendingbeforet2}
\begin{split}
    \sum_{s=1}^{t_2} X'[s]\pmb1\{X[s] < X'[s]\} < t_2\Big((1-\alpha)\mathbb E[b']\\
     + \alpha\mathbb E[b'\pmb1\{b' > r\}]\Big) +  \bar b\sqrt{T\ln T}
\end{split}
\end{equation}
\begin{equation}
\label{eq:utilitybeforet1wedget2}
\begin{split}
    \sum_{s=1}^{t_1\wedge t_2}V[s]\pmb1\{X[s] \geq X'[s]\} < \alpha\Pr(r \geq b')(t_1\wedge t_2)\\
     + \sqrt{T\ln T}
\end{split}
\end{equation}
\begin{equation}
\label{eq:spendingbeforet2}
\begin{split}
    \sum_{s=1}^{t_2}X[s]\pmb1\{X[s]\geq X'[s]\} > \alpha t_2\mathbb E[r\pmb1\{r\geq b'\}]\\
    - \bar b\sqrt{T\ln T}
\end{split}
\end{equation}
\begin{equation}
\label{eq:adversaryspendingbeforet3}
\begin{split}
    \sum_{s=1}^{t_3}X'[s]\pmb1\{X[s]<X'[s]\} > t_3\Big((1-\alpha)\mathbb E[b']\\
     + \alpha\mathbb E[b'\pmb1\{b'<r\}]\Big) -  \bar b\sqrt{T\ln T}
\end{split}
\end{equation}
\fi

\begin{equation}
\label{eq:spendingbetweent3andt4}
    \sum_{s=t_3+1}^{t_4}X[s] < \alpha(t_4 - t_3)\mathbb E[r] +  \bar b\sqrt{T\ln T}
\end{equation}
\begin{equation}
\label{eq:utilitybetweent2andt4}
    \sum_{s=t_2+1}^{t_4}V[s] > \alpha(t_4 - t_2) -  \sqrt{T\ln T}
\end{equation}
By Hoeffding's inequality, $\Pr(E) \geq 1 - O(1/T)$.  (Observe that bids are uniformly bounded by $\bar b$ and the Bernoulli values are uniformly bounded by $1$ so we can indeed apply Hoeffding's inequality like this.)

In what follows, consider what happens on the high probability event $E$. The following hold. By \eqref{eq:spendingbeforet1},

\ifarxiv 

\begin{equation*}
\begin{split}
    t_1 < T \implies\sum_{t=1}^{t_1\wedge \tau}P[t] & \geq \sum_{t=1}^{t_1\wedge\tau}X[t]\pmb1\{X[t]\geq X'[t]\}\\
    & \geq \min\{\alpha T,\alpha t_1\mathbb E[r\pmb1\{r\geq b'\}] - \bar b\sqrt{T\ln T}\}\\
    & \geq \alpha T.
\end{split}
\end{equation*}
By \eqref{eq:adversaryspendingbeforet2},
\begin{equation*}
\begin{split}
    \sum_{t=1}^{t_2\wedge\tau}P'[t] & \leq \sum_{t=1}^{t_2\wedge\tau}X'[t]\pmb1\{X[t] < X'[t]\}\\
    & < t_2((1-\alpha)\mathbb E[b'] + \alpha\mathbb E[b'\pmb1\{b' > r\}]) +  \bar b\sqrt{T\ln T}\\
    & \leq (1-\alpha)T - \bar b.\\
\end{split}
\end{equation*}
By \eqref{eq:adversaryspendingbeforet3},
\begin{equation*}
\begin{split}
    t_3 < T \implies \sum_{t=1}^{t_3\wedge\tau'}P'[t] & \geq \sum_{t=1}^{t_3\wedge\tau'}X'[t]\pmb1\{X[t] < X'[t]\}\\
    & \geq \min\Big\{(1-\alpha)T - \bar b, t_3((1-\alpha)\mathbb E[b'] + \alpha\mathbb E[b'\pmb1\{b'<r\}]) -  \bar b\sqrt{T\ln T}\Big\}\\
    & \geq (1-\alpha)T - \bar b.
\end{split}
\end{equation*}
By \eqref{eq:spendingbeforet2} and \eqref{eq:spendingbetweent3andt4},
\begin{equation*}
\begin{split}
    t_4 < T \implies  \sum_{t=1}^{t_4}P[t] & \geq \sum_{t=1}^{t_2\wedge\tau} P[t] + \sum_{t=t_3+1}^{t_4}P[t]\\
    & \geq \sum_{t=1}^{t_2\wedge\tau}X[s]\pmb1\{X[s]\geq X'[s]\} + \sum_{t=t_3+1}^{t_4\wedge\tau}X[s]\\
    & \geq\min\{\alpha T, \alpha t_2\mathbb E[r\pmb1\{r\geq b'\}] + \alpha(t_4 - t_3)\mathbb E[r]\}\\
    &  \geq \alpha T.
\end{split}
\end{equation*}
The above imply
\begin{equation*}
    \tau \leq t_1\wedge t_4 \quad \text{and} \quad t_2\wedge\tau\leq \tau'\leq t_3.
\end{equation*}
Then, using \eqref{eq:utilitybeforet1wedget2},
\begin{equation*}
\begin{split}
     \sum_{t=1}^{\tau\wedge\tau'} U[t] & \leq \sum_{t=1}^{t_1\wedge t_3}U[t]\\
     & = \sum_{t=1}^{t_1\wedge t_2}U[t] + \sum_{t=t_2+1}^{t_3}U[t]\\
     & \leq \sum_{t=1}^{t_1\wedge t_2}U[t] + O(\sqrt{T\log T})\\
     & = \sum_{t=1}^{t_1\wedge t_2} V[s]\pmb1\{X[s] \geq X'[s]\} + O(\sqrt{T\log T})\\
     & < \alpha\Pr(r \geq b')(t_1\wedge t_2) + O(\sqrt{T\log T})\\
     & \leq \alpha T\Pr(r \geq b')\min\left\{1, \frac{1}{\mathbb E[r\pmb1\{r\geq b'\}]},\frac{1-\alpha}{(1-\alpha)\mathbb E[b'] + \alpha\mathbb E[b'\pmb1\{b' > r\}}\right\} + O(\sqrt{T\log T})
\end{split}
\end{equation*}
and using \eqref{eq:utilitybetweent2andt4},
\begin{equation*}
\begin{split}
    \sum_{t=\tau'+1}^\tau U[t] & = \sum_{t=\tau'+1}^\tau V[t]\\
    & \leq \sum_{t=t_2+1}^{t_4}V[t]\\
    & < \alpha(t_4-t_2)^+ + O(\sqrt{T\log T})\\
    & \leq \alpha T\left(\min\left\{1- \frac{1-\alpha}{(1-\alpha)\mathbb E[b'] + \alpha\mathbb E[b'\pmb1\{b' > r\}]},\frac{1 - \frac{1-\alpha}{(1-\alpha)\mathbb E[b'] + \alpha\mathbb E[b'\pmb1\{b' > r\}]}\mathbb E[r\pmb1\{r\geq b'\}]}{\mathbb E[r]}\right\}\right)^+\\
    & \quad + O(\sqrt{T\log T}).
\end{split}
\end{equation*}

\else

\begin{equation*}
\begin{split}
    t_1 < T & \implies\sum_{t=1}^{t_1\wedge \tau}P[t] \geq \sum_{t=1}^{t_1\wedge\tau}X[t]\pmb1\{X[t]\geq X'[t]\}\\
    & \geq \min\{\alpha T,\alpha t_1\mathbb E[r\pmb1\{r\geq b'\}] - \bar b\sqrt{T\ln T}\}\\
    & \geq \alpha T.
\end{split}
\end{equation*}
By \eqref{eq:adversaryspendingbeforet2},
\begin{equation*}
\begin{split}
    \sum_{t=1}^{t_2\wedge\tau}P'[t] & \leq \sum_{t=1}^{t_2\wedge\tau}X'[t]\pmb1\{X[t] < X'[t]\}\\
    & < t_2((1-\alpha)\mathbb E[b'] + \alpha\mathbb E[b'\pmb1\{b' > r\}])\\
    & \quad +  \bar b\sqrt{T\ln T}\\
    & \leq (1-\alpha)T - \bar b.\\
\end{split}
\end{equation*}
By \eqref{eq:adversaryspendingbeforet3},
\begin{equation*}
\begin{split}
    t_3 < T & \implies \sum_{t=1}^{t_3\wedge\tau'}P'[t]\\
    & \geq \sum_{t=1}^{t_3\wedge\tau'}X'[t]\pmb1\{X[t] < X'[t]\}\\
    & \geq \min\Big\{(1-\alpha)T - \bar b,\\
    & \quad t_3((1-\alpha)\mathbb E[b'] + \alpha\mathbb E[b'\pmb1\{b'<r\}])\\
    & \quad\quad -  \bar b\sqrt{T\ln T}\Big\}\\
    & \geq (1-\alpha)T - \bar b.
\end{split}
\end{equation*}
By \eqref{eq:spendingbeforet2} and \eqref{eq:spendingbetweent3andt4},
\begin{equation*}
\begin{split}
    t_4 & < T \implies  \sum_{t=1}^{t_4}P[t]\geq \sum_{t=1}^{t_2\wedge\tau} P[t] + \sum_{t=t_3+1}^{t_4}P[t]\\
    & \geq \sum_{t=1}^{t_2\wedge\tau}X[s]\pmb1\{X[s]\geq X'[s]\} + \sum_{t=t_3+1}^{t_4\wedge\tau}X[s]\\
    & \geq\min\{\alpha T, \alpha t_2\mathbb E[r\pmb1\{r\geq b'\}] + \alpha(t_4 - t_3)\mathbb E[r]\}\\
    &  \geq \alpha T.
\end{split}
\end{equation*}
The above imply
\begin{equation}
    \tau \leq t_1\wedge t_4 \quad \text{and} \quad t_2\wedge\tau\leq \tau'\leq t_3.
\end{equation}
Then, using \eqref{eq:utilitybeforet1wedget2},
\begin{equation*}
\begin{split}
     \sum_{t=1}^{\tau\wedge\tau'} & U[t] \leq \sum_{t=1}^{t_1\wedge t_3}U[t]\\
     & = \sum_{t=1}^{t_1\wedge t_2}U[t] + \sum_{t=t_2+1}^{t_3}U[t]\\
     & \leq \sum_{t=1}^{t_1\wedge t_2}U[t] + O(\sqrt{T\log T})\\
     & = \sum_{t=1}^{t_1\wedge t_2} V[s]\pmb1\{X[s] \geq X'[s]\} + O(\sqrt{T\log T})\\
     & < \alpha\Pr(r \geq b')(t_1\wedge t_2) + O(\sqrt{T\log T})\\
     & \leq \alpha T\Pr(r \geq b')\min\left\{1, \frac{1}{\mathbb E[r\pmb1\{r\geq b'\}]},\right.\\
     & \left.\quad\frac{1-\alpha}{(1-\alpha)\mathbb E[b'] + \alpha\mathbb E[b'\pmb1\{b' > r\}}\right\}\\
     & \quad + O(\sqrt{T\log T})
\end{split}
\end{equation*}
and using \eqref{eq:utilitybetweent2andt4},
\begin{equation*}
\begin{split}
    & \sum_{t=\tau'+1}^\tau U[t] = \sum_{t=\tau'+1}^\tau V[t]\\
    & \quad \leq \sum_{t=t_2+1}^{t_4}V[t]\\
    & \quad < \alpha(t_4-t_2)^+ + O(\sqrt{T\log T})\\
    & \quad \leq \alpha T\left(\min\left\{1\right.\right.\\
    & \quad\left.\left.\quad - \frac{1-\alpha}{(1-\alpha)\mathbb E[b'] + \alpha\mathbb E[b'\pmb1\{b' > r\}]},\right.\right.\\
    & \quad\left.\left.\quad\frac{1 - \frac{1-\alpha}{(1-\alpha)\mathbb E[b'] + \alpha\mathbb E[b'\pmb1\{b' > r\}]}\mathbb E[r\pmb1\{r\geq b'\}]}{\mathbb E[r]}\right.\right.\\
    & 
    \left.\left.\quad\quad\right\}\right)^+ + O(\sqrt{T\log T}).
\end{split}
\end{equation*}

\fi 

The above gives us an upper bound of the agent's utility on the event $E$. Since $\Pr(E) \geq 1 - O(1/T)$, the above can be seen to imply the lemma statement using \eqref{eq:agentobjective} and Markov's inequality.  (Recall that we have Bernoulli valuations so the values are bounded by $1$ on the complement event $E^c$.)
\end{proof}

We can simplify \eqref{eq:agentmodifiedobjective} with a slightly higher upper bound for small $\alpha$ as follows. Suppose that we make the adversary pay their bid each round no matter who wins as opposed to only making them pay when they wins. The agent should get higher expected utility this way since the adversary is paying more. When the agent still has tokens, the adversary's expected payment goes from $(1-\alpha)\mathbb E[b'] + \alpha\mathbb E[b'\pmb1\{b'>r\}]$ to $\mathbb E[b']$. Doing this replacement in \eqref{eq:agentmodifiedobjective} gives us the following bound on $U(\alpha,\mathcal D,\mathcal D')$.
\begin{lemma}
\label{lem:simplifiedutilitymodifiedgame}
\ifarxiv 
\begin{equation*}
\begin{split}
    \frac{1}{\alpha T} U(\alpha, \mathcal D, \mathcal D') & \leq \Pr(r \geq b')\min\left\{1, \frac{1}{\mathbb E[r\pmb1\{r\geq b'\}]}, \frac{1-\alpha}{\mathbb E[b']}\right\}\\
    & \quad + \left(\min\left\{1-\frac{1-\alpha}{\mathbb E[b']}, \quad \frac{1-\frac{1-\alpha}{\mathbb E[b']}\mathbb E[r\pmb1\{r\geq b'\}]}{\mathbb E[r]}\right\}\right)^+
\end{split}
\end{equation*}
\else
\begin{equation*}
\begin{split}
    \frac{1}{\alpha T}& U(\alpha, \mathcal D, \mathcal D') \leq \Pr(r \geq b')\min\left\{1, \right.\\
    & \left.\frac{1}{\mathbb E[r\pmb1\{r\geq b'\}]}, \frac{1-\alpha}{\mathbb E[b']}\right\}\\
    & + \left(\min\left\{1-\frac{1-\alpha}{\mathbb E[b']},\right.\right.\\
    & \left.\left.\quad \frac{1-\frac{1-\alpha}{\mathbb E[b']}\mathbb E[r\pmb1\{r\geq b'\}]}{\mathbb E[r]}\right\}\right)^+
\end{split}
\end{equation*}
\fi 
\end{lemma}
\begin{proof}
Let
\begin{equation*}
    x = \frac{1-\alpha}{(1-\alpha)\mathbb E[b'] + \alpha\mathbb E[b'\pmb1\{b' > r\}]}.
\end{equation*}
Observe that if $\min\left\{1, \frac{1}{\mathbb E[r\pmb1\{r\geq b'\}]}\right\} \leq x$,
\ifarxiv 
\begin{equation}
\label{eq:xagentrunsoutofmoneyfirst}
\begin{split}
    \frac{1}{\alpha T}U(\alpha, &\mathcal D,\mathcal D') = \Pr(r\geq b')\min\left\{1, \frac{1}{\mathbb E[r\pmb1\{r\geq b'\}]}\right\},
\end{split}
\end{equation}
\else
\begin{equation}
\label{eq:xagentrunsoutofmoneyfirst}
\begin{split}
    \frac{1}{\alpha T}U(\alpha, &\mathcal D,\mathcal D')\\
     &= \Pr(r\geq b')\min\left\{1, \frac{1}{\mathbb E[r\pmb1\{r\geq b'\}]}\right\},
\end{split}
\end{equation}
\fi
and otherwise,
\ifarxiv 
\begin{equation}
\label{eq:xadversaryrunsoutofmoneyfirst}
\begin{split}
    \frac{1}{\alpha T}& U(\alpha,\mathcal D,\mathcal D') = x\Pr(r\geq b') + \min\left\{1-x, \frac{1-x\mathbb E[r\pmb1\{r\geq b'\}]}{\mathbb E[r]}\right\}.
\end{split}
\end{equation}
\else
\begin{equation}
\label{eq:xadversaryrunsoutofmoneyfirst}
\begin{split}
    \frac{1}{\alpha T}& U(\alpha,\mathcal D,\mathcal D') = x\Pr(r\geq b')\\
     & \quad+ \min\left\{1-x, \frac{1-x\mathbb E[r\pmb1\{r\geq b'\}]}{\mathbb E[r]}\right\}.
\end{split}
\end{equation}
\fi
With $\frac{1}{\alpha T}U(\alpha, \mathcal D,\mathcal D')$ expressed this way, the lemma states that we can replace $x$ with $\frac{1-\alpha}{\mathbb E[b']}$ and not make it any smaller. Since $x \geq \frac{1-\alpha}{\mathbb E[b']}$, it suffices to show that the expressed piecewise function of $x$ as given in \eqref{eq:xagentrunsoutofmoneyfirst} and \eqref{eq:xadversaryrunsoutofmoneyfirst} is nonincreasing, for which it suffices to show its left-hand derivative is always nonpositive. This left-hand derivative is $0$ if $\min\left\{1, \frac{1}{\mathbb E[r\pmb1\{r\geq b'\}]}\right\} \leq x$ and otherwise, it is at most
\begin{equation*}
    \Pr(r\geq b') + \max\left\{-1, -\frac{\mathbb E[r\pmb1\{r\geq b'\}}{\mathbb E[r]}\right\}
\end{equation*}
which is always nonpositive by Markov's inequality and the independence of $r$ and $b'$. 

\end{proof}

Let $F$ denote the CDF of $\mathcal D$. Let $F(x-)$ denote the left limit of $F$ at $x$. The below lemma comes from the fact that against adversaries that just bid a deterministic constant $x$ until they run out of tokens, the agent still needs to be bidding high enough to beat the adversary enough times in order not to run out of rounds to spend tokens in.

\cdfrestriction*

\begin{proof}
Set $\mathcal D'$ to be the deterministic constant $x$. By \cref{lem:simplifiedutilitymodifiedgame},
\ifarxiv 
\begin{equation*}
\begin{split}
    \frac{1}{\alpha T}U(\alpha, \mathcal D, \mathcal D') & \leq \frac{(1-\alpha)\Pr(r \geq x)}{x} + \left(1 - \frac{1-\alpha}{x}\right).
\end{split}
\end{equation*}
\else
\begin{equation*}
\begin{split}
    \frac{1}{\alpha T}U(\alpha, \mathcal D, \mathcal D') & \leq \frac{(1-\alpha)\Pr(r \geq x)}{x}\\
    & \quad + \left(1 - \frac{1-\alpha}{x}\right).
\end{split}
\end{equation*}
\fi
This is at least $\beta$, which when substituting $\Pr(r\geq x) = 1 - F(x-)$ and solving for $F(x-)$, implies the result.
\end{proof}

The below lemma comes from the adversary choosing a distribution $\mathcal D'$ to be essentially a constant to make the agent and adversary run out of tokens at the same time. Under these strategies, we bound the probability that the agent can win against the adversary when they are both bidding.
\begin{lemma}
\label{lem:fixeddistributionupperbound}
For any choice of distribution $\mathcal D$,
\begin{equation*}
    \inf_{0 < \alpha \leq 1}\frac{1}{\alpha T}\inf_{\mathcal D'}U(\alpha, \mathcal D, \mathcal D') \leq \frac35.
\end{equation*}
\end{lemma}
\begin{proof}
Fix $\mathcal D$. Let $\epsilon > 0$. We will give a distribution $\mathcal D_\epsilon'$ to further upper bound the bound in  \cref{lem:simplifiedutilitymodifiedgame}. Let $r\sim\mathcal D$. Let
\begin{equation*}
     y = \sup\left\{x:\mathbb E[r\pmb1\{r\geq x\}]- \frac{x}{1-\alpha}\geq 0\right\}.
\end{equation*}
Observe that $x\mapsto \mathbb E[r\pmb1\{r\geq x\}]- \frac{x}{1-\alpha}$ is nonincreasing and left-continuous, so that $\mathbb E[r\geq\pmb1\{r\geq y\}]-\frac{y}{1-\alpha}\geq 0$ and $\mathbb E[r\pmb1\{r\geq y+\epsilon\}]-\frac{y+\epsilon}{1-\alpha}<0$. If $y \geq 1-\alpha$, set $\mathcal D_\epsilon'$ to be the distribution that is $y$ with probability $p$ and $y+\epsilon$ with probability $1-p$ where $p$ is set such that if $b_\epsilon'\sim\mathcal D_\epsilon'$, then $\frac{\mathbb E[b_\epsilon']}{1-\alpha} = \mathbb E[r\geq\pmb1\{r\geq b_\epsilon'\}]$. Otherwise, if $y < 1-\alpha$, set $\mathcal D_\epsilon'$ to be the deterministic constant $1-\alpha$. In this case, $\mathbb E\left[r\pmb1\left\{r\geq 1-\alpha\right\}\right] < 1$, but by changing $\mathcal D$ by increasing $r\sim \mathcal D$ appropriately, we can make $\mathbb E\left[r\pmb1\left\{r\geq 1-\alpha\right\}\right] = 1$ while not decreasing $\Pr(r\geq b_\epsilon')$, which does not decrease the bound of $U(\alpha, \mathcal D, \mathcal D_\epsilon')$ in \cref{lem:simplifiedutilitymodifiedgame}. Without loss of generality, assume that we have modified $\mathcal D$ in this way so that $\mathbb E\left[r\pmb1\left\{r\geq 1-\alpha\right\}\right] = 1$. Then, in all cases, $y\geq 1-\alpha$, and $\mathcal D_\epsilon'$ is the distribution that is $y$ with some probability $p$ and $y+\epsilon$ with probability $1-p$ such that $\frac{\mathbb E[b_\epsilon']}{1-\alpha} = \mathbb E[r\pmb1\{r\geq y\}]$.

 Let $\beta_\epsilon(\alpha) = \frac{1}{\alpha T}U(\alpha, \mathcal D, \mathcal D_\epsilon')$. By \cref{lem:simplifiedutilitymodifiedgame},
\begin{equation*}
    \beta_\epsilon(\alpha) \leq \frac{(1-\alpha)\Pr(r\geq b_\epsilon')}{\mathbb E[b_\epsilon']}.
\end{equation*}
Let $\gamma_\epsilon(\alpha) = \frac{1-\beta_\epsilon(\alpha)}{1-\alpha}$. By \cref{lem:cdfrestriction},
\begin{equation*}
    \Pr(r\geq x) \geq 1-\gamma_\epsilon(\alpha)x
\end{equation*}
for $1-\alpha \leq x \leq \frac{1}{\gamma_\epsilon(\alpha)}$. We calculate
\ifarxiv 
\begin{equation*}
\begin{split}
    \mathbb E & [r\pmb1\{r\geq b_\epsilon'\}] = \mathbb E[\mathbb E[r\pmb1\{r\geq b_\epsilon'\}]\mid b_\epsilon']\\
    & = \mathbb E\left[\mathbb E\left[b_\epsilon'\Pr(r\geq b_\epsilon') + \int_{b_\epsilon'}^\infty \Pr(r\geq x)\,dx\,\middle|\, b_\epsilon'\right]\right]\\
    & \leq \mathbb E\bigg[\mathbb E\bigg[b_\epsilon'\Pr(r\geq b_\epsilon') + \int_{b_\epsilon'}^{1/\gamma_\epsilon(\alpha)}(1 - \gamma_\epsilon(\alpha) x)\,dx \,\bigg|\, b_\epsilon'\bigg]\bigg]\\
    & = \mathbb E[b_\epsilon'\Pr(r\geq b_\epsilon'\mid b_\epsilon')] + \mathbb E\left[\frac{(1 - \gamma_\epsilon(\alpha) b_\epsilon')^2}{2\gamma_\epsilon(\alpha)}\right]\\
    & \leq (y+\epsilon)\Pr(r\geq b_\epsilon') + \frac{(1 - \gamma_\epsilon(\alpha) y)^2}{2\gamma_\epsilon(\alpha)}\\
    & \leq (y+\epsilon)\frac{\beta_\epsilon(\alpha)\mathbb E[b_\epsilon']}{1-\alpha} + \frac{(1 - \gamma_\epsilon(\alpha) y)^2}{2\gamma_\epsilon(\alpha)}\\
    & \leq (y+\epsilon)^2\frac{\beta_\epsilon(\alpha)}{1-\alpha} + \frac{(1 - \gamma_\epsilon(\alpha) y)^2}{2\gamma_\epsilon(\alpha)}\\
    & = (y+\epsilon)^2\left(\frac{1}{1-\alpha} - \gamma_\epsilon(\alpha)\right) + \frac{(1 - \gamma_\epsilon(\alpha) y)^2}{2\gamma_\epsilon(\alpha)}.
\end{split}
\end{equation*}
\else
\begin{equation*}
\begin{split}
    \mathbb E & [r\pmb1\{r\geq b_\epsilon'\}] = \mathbb E[\mathbb E[r\pmb1\{r\geq b_\epsilon'\}]\mid b_\epsilon']\\
    & = \mathbb E\left[\mathbb E\left[b_\epsilon'\Pr(r\geq b_\epsilon') + \int_{b_\epsilon'}^\infty \Pr(r\geq x)\,dx\,\middle|\, b_\epsilon'\right]\right]\\
    & \leq \mathbb E\bigg[\mathbb E\bigg[b_\epsilon'\Pr(r\geq b_\epsilon')\\
    & \quad + \int_{b_\epsilon'}^{1/\gamma_\epsilon(\alpha)}(1 - \gamma_\epsilon(\alpha) x)\,dx \,\bigg|\, b_\epsilon'\bigg]\bigg]\\
    & = \mathbb E[b_\epsilon'\Pr(r\geq b_\epsilon'\mid b_\epsilon')] + \mathbb E\left[\frac{(1 - \gamma_\epsilon(\alpha) b_\epsilon')^2}{2\gamma_\epsilon(\alpha)}\right]\\
    & \leq (y+\epsilon)\Pr(r\geq b_\epsilon') + \frac{(1 - \gamma_\epsilon(\alpha) y)^2}{2\gamma_\epsilon(\alpha)}\\
    & \leq (y+\epsilon)\frac{\beta_\epsilon(\alpha)\mathbb E[b_\epsilon']}{1-\alpha} + \frac{(1 - \gamma_\epsilon(\alpha) y)^2}{2\gamma_\epsilon(\alpha)}\\
    & \leq (y+\epsilon)^2\frac{\beta_\epsilon(\alpha)}{1-\alpha} + \frac{(1 - \gamma_\epsilon(\alpha) y)^2}{2\gamma_\epsilon(\alpha)}\\
    & = (y+\epsilon)^2\left(\frac{1}{1-\alpha} - \gamma_\epsilon(\alpha)\right) + \frac{(1 - \gamma_\epsilon(\alpha) y)^2}{2\gamma_\epsilon(\alpha)}.
\end{split}
\end{equation*}
\fi
Since $\mathbb E[r\pmb1\{r\geq b_\epsilon'\}] \geq \frac{y}{1-\alpha}$, we obtain
\ifarxiv 
\begin{equation*}
\begin{split}
    \frac{y}{1-\alpha} \geq (y+\epsilon)^2\left(\frac{1}{1-\alpha} - \gamma_\epsilon(\alpha)\right) + \frac{(1 - \gamma_\epsilon(\alpha) y)^2}{2\gamma_\epsilon(\alpha)}.
\end{split}
\end{equation*}
\else
\begin{equation*}
\begin{split}
    \frac{y}{1-\alpha} \geq (y+\epsilon)^2\left(\frac{1}{1-\alpha} - \gamma_\epsilon(\alpha)\right)\\
    + \frac{(1 - \gamma_\epsilon(\alpha) y)^2}{2\gamma_\epsilon(\alpha)}.
\end{split}
\end{equation*}
\fi
Solving for $\gamma_\epsilon(\alpha)$ above, we obtain
\ifarxiv 
\begin{equation}
\begin{split}
\label{eq:quadraticingamma}
    (2(y+ & \epsilon)^2 - y^2)\gamma_\epsilon(\alpha)^2 + \frac{2}{1-\alpha}\left(y(2-\alpha) - (y+\epsilon)^2)\right)\gamma_\epsilon(\alpha) - 1 \geq 0.
\end{split}
\end{equation}
\else
\begin{equation}
\begin{split}
\label{eq:quadraticingamma}
    (2(y+ & \epsilon)^2 - y^2)\gamma_\epsilon(\alpha)^2\\
    &\quad+ \frac{2}{1-\alpha}\left(y(2-\alpha) - (y+\epsilon)^2)\right)\gamma_\epsilon(\alpha) - 1\\
    & \quad\quad\quad \geq 0.
\end{split}
\end{equation}
\fi
By Vieta's formulas, the quadratic in $\gamma_\epsilon(\alpha)$ on the left-hand side has a negative root $\gamma_-(\alpha, y,\epsilon)$ and a positive root $\gamma_+(\alpha, y,\epsilon)$. By \eqref{eq:quadraticingamma},
\begin{equation*}
    \gamma_\epsilon(\alpha) \geq \gamma_+(\alpha, y,\epsilon).
\end{equation*}
Using the quadratic formula and some calculus, it can be checked that
\ifarxiv 
\begin{equation*}
\begin{split}
    \inf_{\substack{\epsilon >0\\0 < \alpha\leq 1}} \beta_\epsilon(\alpha) & = \inf_{\substack{\epsilon >0\\0 < \alpha\leq 1}}\left[1-(1-\alpha)\gamma_\epsilon(\alpha)\right]\\
    & \leq \inf_{0< \alpha\leq 1}\left[1-(1-\alpha)\sup_{\epsilon>0}\inf_{y\geq 1-\alpha}\gamma_+(\alpha, y,\epsilon)\right]\\
    & = \sup_{y\geq 1}\lim_{\alpha\to0}\lim_{\epsilon\to0}[1-(1-\alpha)\gamma_+(\alpha, y,\epsilon)]\\
    & = \sup_{y\geq 1}\left[\frac{2-\sqrt{y^2-4y+5}}{y}\right]\\
    & = \frac35.
\end{split}
\end{equation*}
\else
\begin{equation*}
\begin{split}
    \inf_{\substack{\epsilon >0\\0 < \alpha\leq 1}} & \beta_\epsilon(\alpha) = \inf_{\substack{\epsilon >0\\0 < \alpha\leq 1}}\left[1-(1-\alpha)\gamma_\epsilon(\alpha)\right]\\
    & \leq \inf_{0< \alpha\leq 1}\left[1-(1-\alpha)\sup_{\epsilon>0}\inf_{y\geq 1-\alpha}\gamma_+(\alpha, y,\epsilon)\right]\\
    & = \sup_{y\geq 1}\lim_{\alpha\to0}\lim_{\epsilon\to0}[1-(1-\alpha)\gamma_+(\alpha, y,\epsilon)]\\
    & = \sup_{y\geq 1}\left[\frac{2-\sqrt{y^2-4y+5}}{y}\right]\\
    & = \frac35.
\end{split}
\end{equation*}
\fi
\end{proof}

Combining  \cref{lem:agentmodifiedobjective,lem:fixeddistributionupperbound}, we get the proof of \cref{thm:fixeddistributionupperbound}.

\subsection{Deferred proof of Theorem \ref{thm:arbitrarystrategyupperbound}}

In this section we prove the theorem of \cref{ssec:upper:arbitrary_distr}, which we first restate for completeness.

\arbitrarystrategyupperbound*

\begin{proof}
For ease of exposition, we assume the agent wins in case of ties in bidding.
The adversary uses the following bid distribution
\ifarxiv 
\begin{equation*}
\begin{split}
    \Pr_{X'\sim\mathcal D'} & (X'\leq x) = \begin{cases}0 & \text{if $x < 0$}\\\frac{1-\alpha-\delta}{e-(e-2)x} + \alpha + \delta & \text{if $0\leq x \leq \frac{e-1}{e-2}$}\\1 & \text{if $x > \frac{e-1}{e-2}$}\end{cases}.
\end{split}
\end{equation*}
\else
\begin{equation*}
\begin{split}
    \Pr_{X'\sim\mathcal D'} & (X'\leq x)\\
    & = \begin{cases}0 & \text{if $x < 0$}\\\frac{1-\alpha-\delta}{e-(e-2)x} + \alpha + \delta & \text{if $0\leq x \leq \frac{e-1}{e-2}$}\\1 & \text{if $x > \frac{e-1}{e-2}$}\end{cases}.
\end{split}
\end{equation*}
\fi
with $\delta = \sqrt{\frac{3\ln T}{T}}$.
In particular, the adversary bids according to this distribution every round until they have less than $\frac{e-1}{e-2}$ tokens left.

Let $X'[1], X'[2], \dots, X'[T]$ be i.i.d. drawn from $\mathcal D'$ such that $b'[t] = X'[t]\pmb1\left\{B'[t] \geq \frac{e-1}{e-2}\right\}$.  It is easily computed that $\mathbb E[X'[t]] = 1-\alpha-\delta$. By the Chernoff bound,

\ifarxiv 

\begin{equation}
\label{eq:chernoffupperbound}
\begin{split}
    \Pr\left(\sum_{t=1}^T X'[t] \geq (1+\delta)(1-\alpha-\delta)T\right)
    \leq \exp\left(-\frac{\delta^2(1-\alpha-\delta)T}{2+\delta}\right).
\end{split}
\end{equation}
Let $\mathcal H_t$ denote the history up to, and including, time $t$. Let $\mathcal G_t$ be the $\sigma$-algebra generated by $\mathcal H_t$, $V[t+1]$, and $b[t+1]$. Define the $\mathcal G_t$-martingales
\begin{equation*}
\begin{split}
    M_1[t] & = \sum_{s=1}^t b[s]\pmb1\{b[s]\geq X'[s]\} - \sum_{s=1}^t\mathbb E[b[s]\pmb1\{b[s]\geq X'[s]\}\mid b[s]]
\end{split}
\end{equation*}
and
\begin{equation*}
\begin{split}
    M_2[t] & = \sum_{s=1}^t V[s]\pmb1\{b[s]\geq X'[s]\} - \sum_{s=1}^tV[s]\mathbb E[\pmb1\{b[s]\geq X'[s]\}\mid b[s]].
\end{split}
\end{equation*}
Let $\epsilon > 0$. By the Azuma-Hoeffding inequality,
\begin{equation}
\begin{split}
\label{eq:azumaupperbound}
    & \Pr\left(M_1[T] \leq \frac{e-2}{e-1}\epsilon, M_2[T] \geq \epsilon,\sum_{s=1}^T V[s] \geq \alpha T + \epsilon \right)\leq 3\exp\left(-\frac{2\epsilon^2}{T}\right).
\end{split}
\end{equation}

Let $E$ be the high probability event that the events in \eqref{eq:chernoffupperbound} and \eqref{eq:azumaupperbound} do not occur,
\begin{equation}
\label{eq:probabilitygoodeventuniformupperbound}
\begin{split}
    &\Pr(E) \geq 1 - 3\exp\left(-\frac{2\epsilon^2}{T}\right) - \exp\left(-\frac{\delta^2(1-\alpha-\delta)T}{2+\delta}\right).
\end{split}
\end{equation}

\else

\begin{equation}
\label{eq:chernoffupperbound}
\begin{split}
    \Pr\left(\sum_{t=1}^T X'[t] \geq (1+\delta)(1-\alpha-\delta)T\right)\\
    \leq \exp\left(-\frac{\delta^2(1-\alpha-\delta)T}{2+\delta}\right).
\end{split}
\end{equation}
Let $\mathcal H_t$ denote the history up to, and including, time $t$. Let $\mathcal G_t$ be the $\sigma$-algebra generated by $\mathcal H_t$, $V[t+1]$, and $b[t+1]$. Define the $\mathcal G_t$-martingales
\begin{equation*}
\begin{split}
    M_1[t] & = \sum_{s=1}^t b[s]\pmb1\{b[s]\geq X'[s]\}\\
    & \quad - \sum_{s=1}^t\mathbb E[b[s]\pmb1\{b[s]\geq X'[s]\}\mid b[s]]
\end{split}
\end{equation*}
and
\begin{equation*}
\begin{split}
    M_2[t] & = \sum_{s=1}^t V[s]\pmb1\{b[s]\geq X'[s]\}\\
    & \quad - \sum_{s=1}^tV[s]\mathbb E[\pmb1\{b[s]\geq X'[s]\}\mid b[s]].
\end{split}
\end{equation*}
Let $\epsilon > 0$. By the Azuma-Hoeffding inequality,
\begin{equation}
\begin{split}
\label{eq:azumaupperbound}
    & \Pr\left(M_1[T] \leq \frac{e-2}{e-1}\epsilon, M_2[T] \geq \epsilon,\right.\\
    & \left.\quad \sum_{s=1}^T V[s] \geq \alpha T + \epsilon \right)\leq 3\exp\left(-\frac{2\epsilon^2}{T}\right).
\end{split}
\end{equation}

Let $E$ be the high probability event that the events in \eqref{eq:chernoffupperbound} and \eqref{eq:azumaupperbound} do not occur,
\begin{equation}
\label{eq:probabilitygoodeventuniformupperbound}
\begin{split}
    &\Pr(E)\\
    &\geq 1 - 3\exp\left(-\frac{2\epsilon^2}{T}\right) - \exp\left(-\frac{\delta^2(1-\alpha-\delta)T}{2+\delta}\right).
\end{split}
\end{equation}

\fi

Consider what happens on $E$. The adversary never runs out of tokens:
\begin{equation*}
\begin{split}
    \sum_{t=1}^T P'[t] & \leq \sum_{t=1}^T X'[t]\\
    & < (1+\delta)(1-\alpha-\delta)T\\
    & = (1-\alpha)T - \delta(\alpha + \delta)T\\
    & \leq (1-\alpha)T - \frac{e-1}{e-2},
\end{split}
\end{equation*}
for a sufficiently large value of $\delta T$. Hence, we have $P[s] = b[s]\pmb1\{b[s]\geq X'[s]\}$ and $U[s] = V[s]\pmb1\{b[s]\geq X'[s]\}$. The sums of these quantities satisfy, assuming $0\leq b[s] \leq \frac{e-1}{e-2}$ without loss of generality,
\ifarxiv 
\begin{equation*}
\begin{split}
    \sum_{s=1}^T P[s] & = \sum_{s=1}^T \mathbb E [b[s]\pmb1\{b[s]\geq X'[s]\}\mid b[s]]+ M_1[T]\\
    & \geq \sum_{s=1}^T b[s]\left(\frac{1-\alpha-\delta}{e-(e-2)b[s]}+\alpha+\delta\right) - \frac{e-1}{e-2}\epsilon
\end{split}
\end{equation*}
and
\begin{equation*}
\begin{split}
    \sum_{s=1}^T U[s] & = \sum_{s=1}^T V[s]\mathbb E[\pmb1\{b[s]\geq X'[s]\}\mid b[s]] + M_2[T]\\
    & \leq \sum_{s=1}^T V[s]\left(\frac{1-\alpha-\delta}{e-(e-2)b[s]}+\alpha + \delta\right) + \epsilon.
\end{split}
\end{equation*}

\else

\begin{equation*}
\begin{split}
    \sum_{s=1}^T P[s] & = \sum_{s=1}^T \mathbb E [b[s]\pmb1\{b[s]\geq X'[s]\}\mid b[s]]\\
    & \quad + M_1[T]\\
    & \geq \sum_{s=1}^T b[s]\left(\frac{1-\alpha-\delta}{e-(e-2)b[s]}+\alpha+\delta\right)\\
    & \quad - \frac{e-1}{e-2}\epsilon
\end{split}
\end{equation*}
and
\begin{equation*}
\begin{split}
    \sum_{s=1}^T U[s] & = \sum_{s=1}^T V[s]\mathbb E[\pmb1\{b[s]\geq X'[s]\}\mid b[s]]\\
    & \quad + M_2[T]\\
    & \leq \sum_{s=1}^T V[s]\left(\frac{1-\alpha-\delta}{e-(e-2)b[s]}+\alpha + \delta\right) + \epsilon.
\end{split}
\end{equation*}

\fi
Since $\sum_{s=1}^T P[s]\leq \alpha T$, the total utility is bounded by the value of the following maximization problem.
\ifarxiv 
\begin{align*}
    \max_{(b[s])_{s=1}^T}&\sum_{s=1}^T V[s]\left(\frac{1-\alpha-\delta}{e-(e-2)b[s]}+\alpha + \delta\right) + \epsilon \nonumber\\
    \text{s.t.}\quad&\sum_{s=1}^T b[s]\left(\frac{1-\alpha-\delta}{e-(e-2)b[s]}+\alpha+\delta\right) - \frac{e-1}{e-2}\epsilon\leq \alpha T \nonumber\\
    & 0 \leq b[s] \leq \frac{e-1}{e-2} & \forall s \nonumber
\end{align*}
\else
\begin{align*}
    \max_{(b[s])_{s=1}^T}&\sum_{s=1}^T V[s]\left(\frac{1-\alpha-\delta}{e-(e-2)b[s]}+\alpha + \delta\right) + \epsilon \nonumber\\
    \text{s.t.}\quad&\sum_{s=1}^T b[s]\left(\frac{1-\alpha-\delta}{e-(e-2)b[s]}+\alpha+\delta\right)\\
    & \quad\quad\quad\quad\quad\quad\quad\quad\quad\quad- \frac{e-1}{e-2}\epsilon\leq \alpha T \nonumber\\
    & 0 \leq b[s] \leq \frac{e-1}{e-2} & \forall s \nonumber
\end{align*}
\fi
Let $(b^\star[s])_{s=1}^T$ be an optimal solution to the above maximization problem. Observe that the objective function is convex and the feasible solution set is also convex, so we can assume $(b^\star[s])_{s=1}^T$ lies on an extreme point of the feasible set. Furthermore, $b[s]\left(\frac{1-\alpha-\delta}{e-(e-2)b[s]}+\alpha+\delta\right)$ is a nondecreasing function of $b[s]$, so $(b^\star[s])_{s=1}^T$ lying on an extreme point implies that all but one $b^\star[s]$ is either $0$ or $\frac{e-1}{e-2}$. Let $S_0 = \{s\in [T]: b^\star[s]=0\}$ and $S_{(e-1)/(e-2)} = \left\{s\in [T]:b^\star[s]=\frac{e-1}{e-2}\right\}$. By feasibility,
\ifarxiv 
\begin{equation*}
\begin{split}
    \sum_{s\in S_{(e-1)/(e-2)}} b[s]\left(\frac{1-\alpha-\delta}{e-(e-2)b[s]}+\alpha+\delta\right) = |S_{(e-1)/(e-2)}|\frac{e-1}{e-2}\leq \frac{e-1}{e-2}\epsilon + \alpha T.
\end{split}
\end{equation*}
\else
\begin{equation*}
\begin{split}
    \sum_{s\in S_{(e-1)/(e-2)}} b[s]\left(\frac{1-\alpha-\delta}{e-(e-2)b[s]}+\alpha+\delta\right)\\
    = |S_{(e-1)/(e-2)}|\frac{e-1}{e-2}\leq \frac{e-1}{e-2}\epsilon + \alpha T.
\end{split}
\end{equation*}
\fi
We obtain the bound
\begin{equation}
\label{eq:maxbidbound}
    |S_{(e-1)/(e-2)}|\leq \frac{e-2}{e-1}\alpha T + \epsilon.
\end{equation}
The objective function can bounded by
\ifarxiv 
\begin{equation*}
\begin{split}
    & \sum_{s\in S_0}V[s]\left(\frac{1-\alpha-\delta}{e-(e-2)b[s]}+\alpha + \delta\right) + \sum_{s\in S_{(e-1)/(e-2)}}\left(\frac{1-\alpha-\delta}{e-(e-2)b[s]}+\alpha + \delta\right)\\
    & \quad + \sum_{s\notin S_0\cup S_{(e-1)/(e-2)}}\left(\frac{1-\alpha-\delta}{e-(e-2)b[s]}+\alpha + \delta\right) + \epsilon\\
    & \leq \left(\sum_{s=1}^TV[s] - |S_{(e-1)/(e-2)}|\right)\left(\frac{1-\alpha-\delta}{e}+\alpha+\delta\right) + (|S|_{(e-1)/(e-2)} + 1) + \epsilon\\
    & \leq \left(\alpha T + \epsilon - |S_{(e-1)/(e-2)}|\right)\left(\frac{1-\alpha-\delta}{e}+\alpha+\delta\right) + (|S_{(e-1)/(e-2)}| + 1) + \epsilon.
\end{split}
\end{equation*}
\else
\begin{equation*}
\begin{split}
    & \sum_{s\in S_0}V[s]\left(\frac{1-\alpha-\delta}{e-(e-2)b[s]}+\alpha + \delta\right)\\
    & \quad + \sum_{s\in S_{(e-1)/(e-2)}}\left(\frac{1-\alpha-\delta}{e-(e-2)b[s]}+\alpha + \delta\right)\\
    & \quad + \sum_{s\notin S_0\cup S_{(e-1)/(e-2)}}\left(\frac{1-\alpha-\delta}{e-(e-2)b[s]}+\alpha + \delta\right)\\
    & \quad + \epsilon\\
    & \leq \left(\sum_{s=1}^TV[s] - |S_{(e-1)/(e-2)}|\right)\\
    & \quad\quad\left(\frac{1-\alpha-\delta}{e}+\alpha+\delta\right)\\
    & \quad + (|S|_{(e-1)/(e-2)} + 1) + \epsilon\\
    & \leq \left(\alpha T + \epsilon - |S_{(e-1)/(e-2)}|\right)\left(\frac{1-\alpha-\delta}{e}+\alpha+\delta\right)\\
    & \quad + (|S_{(e-1)/(e-2)}| + 1) + \epsilon.
\end{split}
\end{equation*}
\fi
By substituting \eqref{eq:maxbidbound} into the above, we obtain the following bound on the total utility on the event $E$:
\begin{equation*}
    1 + \frac{\alpha T(e - 1 + \alpha + \delta)}{e} + 2\epsilon.
\end{equation*}
The expected total utility is upper bounded by
\ifarxiv 
\begin{equation*}
\begin{split}
    \sum_{t=1}^T \mathbb E[U[t]] \leq & \left(1 + \frac{\alpha T(e - 1 + \alpha + \delta)}{e} + 2\epsilon\right)\Pr(E) + T(1-\Pr(E)).
\end{split}
\end{equation*}
\else
\begin{equation*}
\begin{split}
    \sum_{t=1}^T \mathbb E[U[t]] \leq & \left(1 + \frac{\alpha T(e - 1 + \alpha + \delta)}{e} + 2\epsilon\right)\Pr(E)\\
    & + T(1-\Pr(E)).
\end{split}
\end{equation*}
\fi
Using the above, by setting $\delta = \sqrt{\frac{3\ln T}{T}}$ and $\epsilon = \sqrt{T\ln T}$, and substituting \eqref{eq:probabilitygoodeventuniformupperbound}, we obtain the theorem statement.
\end{proof}

\end{document}